\documentclass[letterpaper,USenglish,cleveref,autoref,thm-restate]{lipics-v2021}

\hideLIPIcs  

\title{Scheduling with Many Shared Resources} 

\author{Max A. Deppert}{Kiel University, Kiel, Germany}{made@informatik.uni-kiel.de}{https://orcid.org/0000-0003-3083-7998}{Supported by the German Research Foundation (DFG) project JA 612/25-1}
\author{Klaus Jansen}{Kiel University, Kiel, Germany}{kj@informatik.uni-kiel.de}{https://orcid.org/0000-0001-8358-6796}{Supported by the German Research Foundation (DFG) project JA 612/25-1}
\author{Marten Maack}{Paderborn University, Paderborn, Germany}{marten.maack@hni.uni-paderborn.de}{https://orcid.org/0000-0001-7918-6642}{Supported by the German Research Foundation (DFG) within the Collaborative Research Centre “On-The-Fly Computing” under the project number 160364472 — SFB 901/3.}
\author{Simon Pukrop}{Paderborn University, Paderborn, Germany}{simonjp@hni.uni-paderborn.de}{https://orcid.org/0000-0002-4473-5215}{Supported by the German Research Foundation (DFG) within the Collaborative Research Centre “On-The-Fly Computing” under the project number 160364472 — SFB 901/3.}
\author{Malin Rau}{Universität Hamburg, Hamburg, Germany}{malin.rau@uni-hamburg.de}{https://orcid.org/0000-0002-5710-560X}{Supported by DFG Research Group ADYN under grant DFG 411362735}

\authorrunning{M.\,A. Deppert, K. Jansen, M. Maack, S. Pukrop, and M. Rau} 

\Copyright{Jane Open Access and Joan R. Public} 



\category{} 

\relatedversion{} 

\nolinenumbers 

\EventEditors{John Q. Open and Joan R. Access}
\EventNoEds{2}
\EventLongTitle{42nd Conference on Very Important Topics (CVIT 2016)}
\EventShortTitle{CVIT 2016}
\EventAcronym{CVIT}
\EventYear{2016}
\EventDate{December 24--27, 2016}
\EventLocation{Little Whinging, United Kingdom}
\EventLogo{}
\SeriesVolume{42}
\ArticleNo{23}
\PassOptionsToPackage{textsize=scriptsize,textwidth=3cm,shadow}{todonotes}
\usepackage{
    caption,subcaption,
    tikz,
    amsmath,mathtools,braket,bm,interval,
    xcolor,
    xspace,
    enumitem,
    todonotes
}
\usepackage[sort&compress,numbers]{natbib}

\usetikzlibrary{decorations.pathreplacing,patterns,external,math, arrows, arrows.meta}

\makeatletter
\g@addto@macro\bfseries{\boldmath}
\g@addto@macro\normalfont{\unboldmath}
\makeatother

\newlist{steps}{enumerate}{1}
\setlist[steps, 1]{wide,labelindent=0pt,parsep=1ex, label ={\textbf{Step \arabic*. }}}

\newcommand\defaulthlines[1]{
  \hlines{#1}{
    {3/2}/$\frac32$/solid/,
    {5/4}/$\frac54$//,
    1/$1$/solid/,
    {3/4}/$\frac34$//,
    {1/2}/$\frac12$//,
    {1/4}/$\frac14$//,
    0/$0$/solid/
  };
}

\newcommand\machines{m}
\newcommand\classes{\mathcal{C}}
\newcommand\resClasses{\bar{\mathcal{C}}}

\newcommand\jobs{\mathcal{J}}
\newcommand\resources{\mathcal{R}}

\newcommand{\Oh}{\mathcal{O}}
\newcommand{\eps}{\varepsilon}
\newcommand{\opt}{\mathrm{OPT}}

\newcommand{\confs}{\mathcal{K}}
\newcommand{\windows}{\mathcal{W}}
\newcommand{\layers}{\Xi}

\newcommand{\problem}{\textnormal{\textsc{MSRS}}\xspace}

\newcommand{\ANoHuge}{\texttt{Algorithm\_no\_huge}\xspace}
\newcommand{\AFiveThird}{\texttt{Algorithm\_5/3}\xspace}
\newcommand{\AThreeHalf}{\texttt{Algorithm\_3/2}\xspace}

\DeclareMathOperator{\poly}{poly}

\DeclarePairedDelimiter\ceil{\lceil}{\rceil}
\DeclarePairedDelimiter\floor{\lfloor}{\rfloor}
\DeclarePairedDelimiter\abs{\lvert}{\rvert}
\DeclarePairedDelimiterX\sett[2]{\lbrace}{\rbrace}{ #1 \,\delimsize| \,\mathopen{} #2 }

\providecommand\given{}
\newcommand\SetSymbol[1][]{%
\nonscript\:#1\vert
\allowbreak
\nonscript\:
\mathopen{}}
\DeclarePairedDelimiterX\Sett[1]\{\}{%
\renewcommand\given{\SetSymbol[\delimsize]}
#1
}

\definecolor{myBlueGreen}{HTML}{00677C}
\definecolor{myVilet}{HTML}{8E217D}
\definecolor{myOrange}{HTML}{F29400}
\definecolor{myRed}{HTML}{E43117}
\definecolor{myGreen}{HTML}{39842E}

\definecolor{mdarkred}{RGB}{165,0,38}
\definecolor{mred}{RGB}{215,25,28}
\definecolor{mdarkorange}{RGB}{244,109,67}
\definecolor{morange}{RGB}{253,174,97}
\definecolor{myellow}{RGB}{254,224,144}
\definecolor{mverylightblue}{RGB}{224,243,248}
\definecolor{mlightblue}{RGB}{171,217,233}
\definecolor{mblue}{RGB}{116,173,209}
\definecolor{mdarkblue}{RGB}{69,117,180}
\definecolor{mverydarkblue}{RGB}{49,54,149}

\theoremstyle{remark}

\newtheorem{stepi}{Step}
\newtheorem{stepii}{Step}
\newtheorem{stepiii}{Step}
\newtheorem*{invariant}{Invariant}

\crefname{observation}{Observation}{Observations}
\crefname{step}{Step}{Steps}
\crefname{stepi}{Step}{Steps}
\crefname{stepii}{Step}{Steps}
\crefname{stepiii}{Step}{Steps}
\makeatletter

\tikzset{drawscheduleapplystyle/.code={\tikzset{#1}}}

\def\schedule@draw<#1>[#2] at(#3) #4;{
  \edef\globalopts{#2}
  \if\relax\detokenize{#1}\relax
    \let\schedule@list\empty
    \foreach \x in {#4} {
      \ifx\schedule@list\empty
        \xdef\schedule@list{\x}
      \else
        \xdef\schedule@list{\x,\schedule@list}
      \fi
    }
  \else
    \xdef\schedule@list{#4}
  \fi
  \xdef\drawscheduleWidth{0.0}
  \begin{scope}[shift={(#3)}]
  \foreach [count=\i from 0] \machine/\w/\machineopts in \schedule@list {
    \edef\newWidth{\drawscheduleWidth+\w}
    \xdef\drawscheduleTime{0.0}
    \ifdim \w pt > 0 pt
      \foreach \t/\labl/\jobopts/\nodeopts in \machine {
        \edef\newTime{\drawscheduleTime+\t}
        \if\relax\detokenize{#1}\relax
          \draw [drawscheduleapplystyle/.expand once=\globalopts,drawscheduleapplystyle/.expand once=\machineopts,drawscheduleapplystyle/.expand once=\jobopts] (\drawscheduleTime,\drawscheduleWidth) rectangle (\newTime,\newWidth) node[pos=.5,drawscheduleapplystyle/.expand once=\nodeopts] {\labl};
        \else
          \draw [drawscheduleapplystyle/.expand once=\globalopts,drawscheduleapplystyle/.expand once=\machineopts,drawscheduleapplystyle/.expand once=\jobopts] (\drawscheduleWidth,\drawscheduleTime) rectangle (\newWidth,\newTime) node[pos=.5,drawscheduleapplystyle/.expand once=\nodeopts] {\labl};
        \fi
        \xdef\drawscheduleTime{\newTime}
      }
      \xdef\drawscheduleWidth{\newWidth}
    \else
      \xdef\drawscheduleWidth{\drawscheduleWidth-\w}
    \fi
  }
  \end{scope}
}

\newcommand\schedule{}
\def\schedule{\schedule@checkorient}
\def\schedule@checkorient{\pgfutil@ifnextchar u{\schedule@up}{\schedule@checkops<>}}
\def\schedule@up up#1;{\schedule@checkops<u>#1;}
\def\schedule@checkops<#1>{\pgfutil@ifnextchar[{\schedule@ops<#1>}{\schedule@checkat<#1>[]}}
\def\schedule@checkat<#1>[#2]{\pgfutil@ifnextchar a{\schedule@at<#1>[#2]}{\schedule@draw<#1>[#2] at(0,0) }}
\def\schedule@ops<#1>[#2]#3;{\schedule@checkat<#1>[#2]#3;}
\def\schedule@at<#1>[#2]#3at#4(#5) #6;{\schedule@draw<#1>[#2] at(#5) #6;}

\newcommand\schedule@drawlines[4][]{
  \foreach \t/\labl/\drawopts/\nodeopts in {#4} {
    \if\relax\detokenize{#1}\relax
      \draw [dash pattern={on 1.5pt off 0.8pt},drawscheduleapplystyle/.expand once=\drawopts] (\t,#3+#2) -- (\t,-#2) node [below,drawscheduleapplystyle/.expand once=\nodeopts] {\labl};
    \else
      \draw [dash pattern={on 1.5pt off 0.8pt},drawscheduleapplystyle/.expand once=\drawopts] (-#2,\t) -- (#3+#2,\t) node [right,drawscheduleapplystyle/.expand once=\nodeopts] {\labl};
    \fi
  }
}

\newcommand\vlines[3][.2]{\schedule@drawlines{#1}{#2}{#3}}
\newcommand\hlines[3][.2]{\schedule@drawlines[h]{#1}{#2}{#3}}

\makeatother

\begin{document}

\maketitle

\begin{abstract}
Consider the many shared resource scheduling problem where jobs have to be scheduled on identical parallel machines with the goal of minimizing the makespan.
However, each job needs exactly one additional shared resource in order to be executed and hence prevents the execution of jobs that need the same resource while being processed.
Previously a $(2m/(m+1))$-approximation was the best known result for this problem.
Furthermore, a $6/5$-approximation for the case with only two machines was known as well as a PTAS for the case with a constant number of machines.
We present a simple and fast 5/3-approximation and a much more involved but still reasonable 1.5-approximation.
Furthermore, we provide a PTAS for the case with only a constant number of machines, which is arguably simpler and faster than the previously known one, as well as a PTAS with resource augmentation for the general case.
The approximation schemes make use of the N-fold integer programming machinery, which has found more and more applications in the field of scheduling recently.
It is plausible that the latter results can be improved and extended to more general cases.
Lastly, we give a $5/4 - \varepsilon$ inapproximability result for the natural problem extension where each job may need up to a constant number (in particular $3$) of different resources.
\end{abstract}


\section{Introduction}

We consider the problem of makespan minimization on identical parallel machines with many shared resources, or many shared resources scheduling (\problem) for short. 
In this problem, we are given $m$ identical machines, a set $\jobs$ of $n$ jobs, and a processing time or size $p_j\in \mathbb{N}_{\geq 0}$ for each job $j\in\jobs$.
Furthermore, each job needs exactly one additional shared resource in order to be executed and no other job needing the same resource can be processed at the same time.
Hence, the jobs are partitioned into (non-empty) classes $\classes$, i.e., $\bigcup\classes = \jobs$, such that each class corresponds to one of the resources.
A schedule $(\sigma, t)$ maps each job to a machine $\sigma: \jobs \rightarrow \Sett{1,\dots,\machines}$ and a starting time $t: \jobs \rightarrow \mathbb{N}_{\geq 0}$.
It is called valid if no two jobs overlap on the same machine and no two jobs of the same class are processed in parallel, i.e.:
\begin{itemize}
    \item $\forall j, j' \in \jobs, j \neq j'$ with $\sigma(j) = \sigma(j')$: $t(j) + p_j \leq t(j')$ or $t(j') + p_j' \leq t(j)$
    \item $\forall c \in \classes:  j, j' \in c, j \neq j'$: $t(j) + p_j \leq t(j')$ or $t(j') + p_j' \leq t(j)$
\end{itemize}
The makespan $C_{\max}$ of a schedule is defined as $ \max_{j\in\jobs}t(j) + p_j$ and the goal is to find a schedule with minimum makespan.
Note that \problem also models the case in which some jobs do not need a resource since in this case private resources can be introduced.

\subparagraph{State of the Art and Motivation.}

The study of scheduling problems with additional resources has a long and rich tradition.
Already in 1983, Blazewicz et al. \cite{DBLP:journals/dam/BlazewiczLK83} provided a classification for such problems along with basic hardness results and several additional surveys have been published since then \cite{DBLP:journals/eor/EdisOO13,Blazewicz2019,DBLP:reference/crc/BlazewiczBF04}.
The \problem problem, in particular, was introduced by Hebrard et al. \cite{DBLP:journals/dam/HebrardHJMPV16} who considered the scheduling of download plans for Earth observation satellites and provided a $(2m/(m+1))$-approximation for the problem.
Strusevich \cite{doi:10.1080/01605682.2020.1772019} revisited \problem and presented an additional application in human resource management.
Moreover, he provided a faster, alternative $(2m/(m+1))$-approximation that is claimed to be simpler as well, and a $6/5$-approximation for the case with only two machines.
The work also extends the three field notation for scheduling problems based on the convention for additional resources introduced in \cite{DBLP:journals/dam/BlazewiczLK83} to encompass the problem at hand.
In particular, \problem is denoted as $P|res\cdot111|C_{\max}$ in this notation.
The most recent result regarding \problem is due to Dósa et al. \cite{DBLP:journals/tcs/DosaKT19} who provided an efficient polynomial time approximation scheme (EPTAS) for \problem with a constant number of machines.
In fact, the EPTAS even works for a more general setting where each job $j$ additionally may only be assigned to a machine belonging to a given set $\mathcal{M}(j)$ of eligible machines. 

We employ standard notation regarding approximation schemes:
A \emph{polynomial time approximation scheme} (PTAS) provides a polynomial time $(1+\eps)$-approximation for each $\eps > 0$.
It is called \emph{efficient}, or EPTAS, if its running time is of the form $f(1/\eps)\poly(|I|)$ where $f$ is some function and $|I|$ the encoding length of the instance $I$.
Moreover, an EPTAS is called \emph{fully polynomial time approximation scheme} (FPTAS) if the function $f$ is a polynomial.

Since \problem includes makespan minimization on identical machines (without resource constraints) as a subproblem, it is NP-hard already on two machines and strongly NP-hard if the number of machines is part of the input due to straightforward reductions from the partition and 3-partition problem, respectively.
Hence, approximation schemes are essentially the best we can hope for.

The \problem problem has also been considered with regard to the total completion time objective \cite{DBLP:journals/corr/abs-1809-05009,DBLP:phd/basesearch/Janssen19}.
The study of this variant is motivated by a scheduling problem in the semiconductor industry.
On one hand, the authors show NP-hardness for generalizations of the problem, and on the other, they argue that the approach yielding a polynomial time algorithm for total completion time minimization in the absence of resource constraints leads to a $(2 - 1/m)$-approximation for the considered problem.

Another way of looking at \problem is to consider it as variant of scheduling with conflicts, where a conflict graph is given in which the jobs are the vertices and no two jobs connected by an edge may be processed at the same time.
This problem was introduced for unit processing times by Baker and Coffman in 1996 \cite{DBLP:journals/tcs/BakerC96}.
It is known to be APX-hard \cite{DBLP:journals/scheduling/EvenHKR09} already on two machines with job sizes at most 4 and a bipartite agreement graph, i.e., the complement of the conflict graph.
There are many positive and negative results for different versions of this problem (see, e.g., \cite{DBLP:journals/tcs/BakerC96,DBLP:journals/scheduling/EvenHKR09} and the references therein).
For instance, the problem is NP-hard on cographs with unit-size jobs but polynomial time solvable if the number of machines is constant \cite{DBLP:conf/mfcs/BodlaenderJ93}.
Note that in the case of \problem, we have a particularly simple cograph, i.e., a collection of disjoint cliques.

\subparagraph{Results.}

We present a $5/3$-approximation in \cref{sec:53}, a $3/2$-approximation in \cref{sec:32}, approximation schemes in \cref{sec:schemes}, and inapproximability results in \cref{sec:inapprox}.
Note that the $5/3$- and $3/2$-approximation have better approximation ratios than the previously known $(2m/(m+1))$-approximation already for 6 and 4 machines, respectively.

The $5/3$-approximation is a simple and fast algorithm that is based on placing full classes of jobs taking special care of classes containing jobs with particularly big sizes and of classes with large processing time overall.
While the $3/2$-approximation reuses some of the ideas and observations of the first result, it is much more involved. 
To achieve the second result, we first design a $3/2$-approximation for the instances in which jobs cannot be too large relative to the optimal makespan and then design an algorithm that carefully places classes containing such large jobs and uses the first algorithm as a subroutine for the placement of the remaining classes.
Note that our approaches are very different to the one in \cite{DBLP:journals/dam/HebrardHJMPV16}, which successively chooses jobs based on their size and the size of the remaining jobs in their class and then inserts them with some procedure designed to avoid resource conflicts, and the one in \cite{doi:10.1080/01605682.2020.1772019}, which merges the classes into single jobs to avoid resource conflicts.

We provide an EPTAS for the variant of \problem where the number of machines is constant and an EPTAS with resource augmentation for the general case.
In particular, we need $\floor{(1+\eps)m}$ many machines in the latter result.
Both results make use of the basic framework introduced in \cite{JKMR21} which in turn utilizes relatively recent algorithmic results for integer programs (IPs) of a particular form -- so-called N-fold IPs.
Compared to the mentioned work by Dósa et al. \cite{DBLP:journals/tcs/DosaKT19} -- which provides an EPTAS for the case with a constant number of machines as well -- our result is arguably simpler and faster (going from at least triply exponential in $m/\eps$ to doubly exponential).
We also provide the result with resource augmentation for the general case, which may be refined in the future to work without resource augmentation as well.
Moreover, it seems plausible that the use of N-fold IPs in the context of scheduling with additional resources may lead to further results in the future, which do not have to be limited to approximation schemes.

Finally, we provide inapproximability results for variants of \problem where each job may need more than one resource.
In particular, we show that there is no better than $5/4$-approximation for the variant of \problem with multiple resources per job, unless $P = NP$, even if no job needs more than three resources and all jobs have processing time 1, 2 or 3. 
Previously, the APX-hardness result due to Even et al. \cite{DBLP:journals/scheduling/EvenHKR09} for scheduling with conflicts was known, which did focus on a different context and in particular does not provide bounds regarding the number of resources a job may require.

\subparagraph{Further Related Work.}

As mentioned above, there exists extensive research regarding scheduling with additional resources and we refer to the surveys \cite{DBLP:journals/dam/BlazewiczLK83,DBLP:journals/eor/EdisOO13,Blazewicz2019,DBLP:reference/crc/BlazewiczBF04} for an overview.
For instance, the variant with only one additional shared renewable resource where each job needs some fraction of the resource capacity has received a lot of attention (see \cite{DBLP:conf/spaa/KlingMRS17,DBLP:conf/esa/JansenR21,DBLP:journals/algorithmica/NiemeierW15,DBLP:journals/talg/JansenMR19} for some relatively recent examples).
Interestingly, Hebrard \cite{DBLP:journals/dam/HebrardHJMPV16} pointed out that this basic setting is more closely related \problem than it first appears: 
Consider the case that we have dedicated machines, i.e., each job is already assigned to a machine and we only have to choose the starting times, each job needs one unit of the singly additional shared resource, and the shared resource has some integer capacity.
This problem is equivalent to \problem if the multiple resources taken on the roles of the machines and the machines take the role of the single resource.
Hence, results for variants of this setting translate to \problem as well.
For instance, \problem can be solved in polynomial time if at most two classes include more than one job \cite{DBLP:journals/dam/KellererS03} and \cite{DBLP:journals/disopt/GrigorievU09} yields a $(3+\eps)$-approximation.

Scheduling with conflicts has also been studied from the orthogonal perspective, where jobs that are in conflict may not be processed on the same \emph{machines}.
This problem was already studied in the 1990's (see e.g. \cite{DBLP:journals/dam/BodlaenderJW94,DBLP:conf/mfcs/BodlaenderJ93}), and there has been a series of recent results \cite{DBLP:conf/esa/DasW17,DBLP:conf/spaa/GrageJK19,DBLP:journals/tcs/PageS20} regarding the setting corresponding to \problem where the conflict graph is a collection of disjoint cliques.

\subparagraph{Preliminaries.}

We introduce some additional notation, and a first observation that will be used throughout the following sections.

For any set of jobs $X$ let $p(X) = \sum_{j\in X}p_j$ denote its total processing time. 
Also let $p(j) = p_j$ for all jobs $j\in\jobs$.
While creating or discussing a schedule, for any machine $m$ denote by $p(m)$ the (current) total load of jobs on that machine $m$.
Subsequently, for a set of machines $M$, $p(M) = \sum_{m\in M} p(m)$.

For any combination of a set $X \in \set{\jobs,\classes}$, a relation $* \in \set{<,\leq,\geq,>}$, and a number $\lambda$, we define $X_{*\lambda} = \set{x \in X| p(x) * \lambda}$. Furthermore, given an interval $v$ let $X_v = \set{x \in X| p(x) \in v}$. For example it holds that $\jobs_{>1/2} = \set{j \in \jobs| p(j) > 1/2}$ and $\classes_{(1/2,3/4]} = \set{c \in \classes|p(c) \in (1/2,3/4]}$.

\begin{note}\label{lower-bounds}
	It holds that $\mathrm{OPT} \geq \max\set{\frac{p(\jobs)}{\machines} ,\max_{c \in \classes}p(c)}$.
\end{note}
Hence, we assume that $\machines < \abs{\classes}$ as otherwise there is a trivial schedule with one machine per class. 
Furthermore, let us assume that we sort the jobs in decreasing order of processing time. 
Consider the jobs $j_m$ and $j_{m+1}$ at position $m$ and $m+1$.
Note that it has to hold that $\mathrm{OPT} \geq p(j_m) + p(j_{m+1})$, since either $j_{m+1}$ has to be scheduled on the same machine as one of the first $m$ jobs, or two of the first $m$ jobs have to be scheduled at the same machine.

\section{A 5/3-approximation}\label{sec:53}

In this section we introduce a first simple algorithm that gives some intuition on the problem that will be used more cleverly in the next section.
We start by lower bounding the makespan $T$ of an optimal schedule and construct a schedule with makespan at most $\frac53 T$. 
The algorithm works by placing full classes of jobs in a specific order.
More precisely, first classes that contain a job  of size at least $\frac12 T$, then classes with total processing time larger than $\frac23 T$, and lastly all residual classes get placed.

\begin{theorem}\label{thm:5-3}
	There exists an algorithm that, for any instance $I$ of \problem, finds a schedule with makespan bounded by $\frac53 T$ in $\Oh(\abs{I})$ steps, where for the jobs $j_m$ and $j_{m+1}$ with $m$-th and $(m+1)$-st largest processing time we define
	$T := \max\set{\frac{1}{\machines}p(\jobs),\max_{c \in \classes}p(c), p(j_m) +p(j_{m+1})}$.
\end{theorem}

As noted earlier, $T$ denotes a lower bound on the makespan.
We scale each job by $1/T$. 
As a consequence all jobs have a processing time in $(0,1]$ and the total load is bounded by $m$. 
Denote by $\classes_{B^+} := \set{c \in \classes | \abs{c \cap \jobs_{>1/2}} = 1}$ all classes containing a job of size greater than $1/2$.
We aim to find a schedule with makespan in $[1,5/3]$.
The following two observations are directly implied by the definition of $T$.

\begin{observation}
\label{obs:halfitems}
	For each class $c \in \classes$ it holds that $\abs{c \cap \jobs_{>1/2}} \leq 1$.
\end{observation}

\begin{observation}
\label{obs:machineBound1}
	It holds that $\abs{\classes_{B^+}} = \abs{\jobs_{>1/2}} \leq m$.
\end{observation}

Lastly, we address classes with a large total processing time.

\begin{lemma}\label{5over3-splitting-argument}
	Each class $c \in \classes_{>2/3} \setminus \classes_{B^+}$
	can be partitioned into parts $c_{1}$ and $c_{2} = c \setminus c_{1}$ such that $1/3 \leq p(c_{1}) \leq 2/3$ and $p(c_{2}) \leq 2/3$. 
	This partition can be found in time $\Oh(\abs{c})$.
\end{lemma}
\begin{proof}
	If there exists a job $j_{\top}$ in $c$ with $p(j_{\top}) >1/3$, we define $c_{1} = \set{j_{\top}}$ and $c_2 = c \setminus c_1$. Note that $c$ does not contain a job with processing time larger than $1/2$ and hence, $p(c_1) \in (1/3,1/2] $ and $p(c_2) = p(c) - p(c_1) < 1-1/3 = 2/3$.

	Otherwise, greedily add jobs from $c$ to an empty set $c_{1}$ until $p(c_{1}) \geq 1/3$ and set $c_2 = c \setminus c_1$.
	Since all the jobs of $c$ have processing time at most $1/3$, it holds that $p(c_{1}) \in [1/3,2/3]$.
	Consequently, it holds that $p(c_{2}) \leq 2/3$ as well.
\end{proof}

\subparagraph*{Algorithm: \AFiveThird}
\begin{stepi}
Consider all classes containing a job with processing time larger than $1/2$, $\classes_{B^+}$. 
Each of these classes is assigned to an individual machine, and all jobs from such a class are scheduled consecutively, see \cref{fig:five-over-three-step1}. 
\end{stepi}
\begin{stepi}
Consider all remaining classes with total processing time larger than $2/3$, $\classes_{> 2/3}\setminus\classes_{B^+}$.
Try to add these classes on the machines filled with the classes $\classes_{B^+}$ and afterward proceeds to empty machines, see \cref{fig:five-over-three-step2}.
If the considered machine has load in $(1,5/3]$, \emph{close} the machine and no longer attempt to place any other job on it. 
Note that after placing the classes $\classes_{B^+}$ all machines remained open.
Let $m_i$ be the machine we try to place class $c \in \classes_{> 2/3}\setminus\classes_{B^+}$ on.
If $m_i$ has load $p(m_i) \leq 5/3 - p(c)$, place the entire class on this machine and close it.  
Otherwise, partition the class $c$ in two parts $c_{1}$ and $c_{2}$ such that $p(c_{2}) \leq p(c_{1}) \leq 2/3$ (cf. \cref{5over3-splitting-argument}).
Place the larger part $c_1$ on the current machine starting at $5/3 - p(c_1)$ and close it, moving to the next machine.
All jobs on this machine are delayed such that the first job starts at $p(c_2)$. 
All jobs from $c_2$ are scheduled between $0$ and $p(c_2)$ on this machine.
If it has load of at least $1$, this machine is closed as well.

\end{stepi}
\begin{stepi}[Greedy]
Finally, place the classes $\classes_{\leq 2/3}\setminus\classes_{B^+}$, see \cref{fig:five-over-three-step3}.
Consider the residual machines one after another and add each class $c \in \classes_{\leq 2/3}\setminus\classes_{B^+}$ entirely to the considered machine.
As soon as the load of a machine exceeds $1$ close it and move to the next.
\end{stepi}

\begin{figure}[ht]
	\centering
	\begin{subfigure}{0.32\textwidth}
		\centering
		\tikzset{external/export next=false}
		\tikz[xscale=0.6,yscale=3]{
\hlines{6}{
    {5/3}/$\frac53$/solid/,
    {4/3}/$\frac43$//,
    1/$1$/solid/,
    {2/3}/$\frac23$//,
    {1/3}/$\frac13$//,
    0/$0$/solid/
};

\schedule up [fill=lightgray] {
    {
        .75/$J_1$//,
        .17///
    }/1/,
    {
        .55/$J_2$//,
        .10///
    }/1/,
    {
        .75/$J_3$//,
        .12///
    }/1/,
    {
        .6/$J_4$//,
        .3///
    }/1/,
    {
        .57/$J_5$//,
        .12///
    }/1/
};
}
		\caption{Classes with large jobs}
		\label{fig:five-over-three-step1}
	\end{subfigure}
	\begin{subfigure}{0.32\textwidth}
		\centering
		\tikzset{external/export next=false}
		\tikz[xscale=0.6,yscale=3]{
\hlines{6}{
    {5/3}/$\frac53$/solid/,
    {4/3}/$\frac43$//,
    1/$1$/solid/,
    {2/3}/$\frac23$//,
    {1/3}/$\frac13$//,
    0/$0$/solid/
};

\schedule up [fill=lightgray] {
    {
        .75/$J_1$//,
        .17///,
        {.08+2/3-.40}//opacity=0/,
        .40//fill=morange/
    }/1/,
    {
        .30//fill=morange/,
        .55/$J_2$//,
        .10///,
        {.05+2/3-.50}//opacity=0/,
        .50//fill=mred/
    }/1/,
    {
        .18//fill=mred/,
        .75/$J_3$//,
        .12///
    }/1/,
    {
        .6/$J_4$//,
        .3///,
        {.1+2/3-.44}//opacity=0/,
        .44//fill=mblue/
    }/1/,
    {
        .26//fill=mblue,
        .57/$J_5$//,
        .12///
    }/1/
};
}
		\caption{Placing large classes}
		\label{fig:five-over-three-step2}
	\end{subfigure}
	\begin{subfigure}{0.32\textwidth}
		\centering
		\tikzset{external/export next=false}
		\tikz[xscale=0.6,yscale=3]{
\hlines{6}{
    {5/3}/$\frac53$/solid/,
    {4/3}/$\frac43$//,
    1/$1$/solid/,
    {2/3}/$\frac23$//,
    {1/3}/$\frac13$//,
    0/$0$/solid/
};

\schedule up [fill=lightgray] {
    {
        .75/$J_1$//,
        .17///,
        {.08+2/3-.40}//opacity=0/,
        .40//fill=morange/
    }/1/,
    {
        .30//fill=morange/,
        .55/$J_2$//,
        .10///,
        {.05+2/3-.50}//opacity=0/,
        .50//fill=mred/
    }/1/,
    {
        .18//fill=mred/,
        .75/$J_3$//,
        .12///
    }/1/,
    {
        .6/$J_4$//,
        .3///,
        {.1+2/3-.44}//opacity=0/,
        .44//fill=mblue/
    }/1/,
    {
        .26//fill=mblue,
        .57/$J_5$//,
        .12///,
        .34//fill=myellow/
    }/1/,
    {
        .23///,
        .5///,
        .04///,
        .17//
    }/1/{fill=myellow}
};
}
		\caption{Adding all other classes}
		\label{fig:five-over-three-step3}
	\end{subfigure}
	\caption{The three steps of the algorithm (where $\jobs_{>1/2} = \set{J_1,\dots,J_5}$)}
	\label{fig:five-over-three}
\end{figure}

\subparagraph*{Algorithm Correctness.}

\begin{lemma}
Given any instance $I = (m,\classes)$ of \problem,  \AFiveThird produces a feasible schedule with makespan at most $\frac{5}{3}\opt(I)$.
\end{lemma}

\begin{proof}
To prove the correctness and approximation ratio of the algorithm, we have to prove the following points:

\begin{itemize}
	\item All jobs can be scheduled
	\item The processing times of two jobs from the same class never overlap.
	\item The latest completion time of a job is given by $5/3$
\end{itemize}

We start by proving that all jobs are scheduled, by showing that the algorithm closes only machines that have a total load of at least $1$. 
Since the total load of the jobs is bounded by $m$, when attempting to schedule the last class, there has to exist a non closed machine.
The only time the algorithm potentially closes a machine with load less than $1$ is in step 2 when a class $\classes_{> 2/3}\setminus\classes_{B^+}$ is split into two parts. 
Let $c_{B^+}$ be the class already on the machine and $c_1$ and $c_2$ be the parts of the class the algorithm tries to schedule in this step, such that $p(c_1) \geq p(c_2)$.
Since the class was split in two by the algorithm it holds that $p(c_{B^+}) + p(c_{1}) + p(c_{2}) > 5/3$. 
Furthermore, since $p(c_{1}) + p(c_{2}) \leq 1$ and $p(c_{1}) \geq p(c_{2})$ it holds that $p(c_{2}) \leq 1/2$ and hence $p(c_{B^+}) + p(c_{1}) > 7/6$.
Hence that closed machine has a load of at least $1$.

Next, we prove that the processing of two jobs from the same class never overlaps in time.
Again, the only time one class is scheduled on more than one machine is step 2. 
When placing the two parts these parts do not overlap, since they have a processing time of at most $1$ and one of the parts starts at $0$ while the other ends at $5/3$.
The algorithm does not generate any overlapping by shifting jobs already on the machine, since those have to originate from classes in $\classes_{B^+}$, which each got placed on an individual machine.

Finally, we prove that the latest completion time of a job is given by $5/3$.
After step 1 all the machines have a load of at most $1$, since each class has a total processing time of at most $1$.
In step 2, we only add an entire class if the total load is bounded by $5/3$. 
If a class is split, the part that is added has a total processing time of at most $2/3$. 
Since before adding this part the machine had a load of at most $1$, the load of the closed machine is bounded by $5/3$.
This concludes the proof of \Cref{thm:5-3}.
\end{proof}

The existence and correctness of \AFiveThird proofs \cref{thm:5-3}.

\section{A 3/2-approximation}\label{sec:32}

In this section we introduce the more involved algorithm hinted at earlier.
While the general idea is similar, finding a lower bound $T$ for the makespan and then placing classes depending on included big jobs and total processing time, the steps are a lot more granular.
We first give a $3/2$-approximation algorithm for instances without jobs of size bigger than $3/4T$.
After that we introduce a second $3/2$-approximation algorithm that places classes with jobs of size bigger than $3/4T$ on distinct machines and fills them with other jobs in a clever way such that we can reuse the first algorithm for the remaining classes.

\begin{theorem}
\label{thm:ThreeHalfAlgorithm}
	There exists an algorithm that for any given instance $I$ of \problem finds a schedule with makespan bounded by $\frac32 \opt$ in $\Oh(n + m \log(m))$ steps.
\end{theorem}

In the following let us assume that we have scaled the instance such that $\opt = 1$.
In order to provide a $3/2$-approximation algorithm, we consider four different types of jobs. 
We split the jobs of a given instance into \emph{huge} jobs $\jobs_H = \jobs_{>3/4} = \set{j \in \jobs\mid p_j> 3/4}$, \emph{big} jobs $\jobs_B = \jobs_{(1/2,3/4]} = \set{j \in \jobs\mid p_j \in (1/2,3/4]}$, \emph{medium} jobs $\jobs_M = \jobs_{(1/4,1/2]} = \set{j \in \jobs\mid p_j \in  (1/4,1/2]}$, and all residual jobs (with a processing time of at most $1/4$) which we refer to as \emph{small} jobs.

Furthermore, turning to the classes $\classes$ we define the subset $\classes_H = \set{c \in \classes : \abs{\jobs_H \cap c} = 1}$ of all classes containing a huge job, the subset $\classes_B = \set{c \in \classes : \abs{\jobs_B \cap c} = 1}$ of all classes containing a big job, the subset $\classes_{\geq 3/4} = \set{c \in \classes | p(c)\geq 3/4}$ of all classes with a total processing time of at least $3/4$, and the subset $\classes_{(1/2,3/4)} = \set{c \in \classes | p(c)\in (1/2,3/4)}$ of all classes with a total processing time in $(1/2,3/4)$.

\begin{lemma}
	\label{lem:available machines}
	For any normalized optimal schedule and the corresponding partition of $\classes$ into $\classes_H, \mathcal{C}_B, \classes_{\geq3/4} \setminus (\classes_H \cup \classes_B)$ and $\classes \setminus \classes_{\geq3/4}$ it holds that
	\[\abs{\classes_H} + \max\Set{\abs{\mathcal{C}_B} , \ceil*{\frac12\left(\abs{\classes_B} + \abs{\classes_{\geq3/4} \setminus (\classes_H \cup \classes_B)}\right)}} \leq m.\]
\end{lemma}
\begin{proof}
	Clearly, it holds that $\abs{\classes_H}+\abs{\classes_B} \leq m$. 
	
	Let us consider the total load processed in the time corridor between $1/4$ and $3/4$ (over the entire schedule).
	For each class $c\in\classes_H$ we have to schedule at least load $1/2$ in this corridor, since the tallest job in $c$, which has a processing time of at least $3/4$, has to start before $1/4$ and has to end after $3/4$.
	For each class $c \in \classes_B$, at least load $1/4$ is scheduled in this corridor since its big job, which has a processing time in $(1/2,3/4)$, has to end after $1/2$ and has to start before $1/2$.
	Finally, each class in $\classes_{\geq3/4} \setminus (\classes_H \cup \classes_B)$ has load of at least $3/4$. Since at most $1/2$ of this load can be scheduled outside of the corridor, there has to be load of at least $1/4$ scheduled inside of this corridor.
	Hence the total load scheduled in this corridor is at least $\frac12\abs{\classes_H} + \frac{1}{4}(\abs{\classes_B} + \abs{\classes_{\geq3/4} \setminus (\classes_H \cup \classes_B)})$.
	
	Since each machine covers at most $1/2$ of this load, it holds that
	\[
	    m
	    \geq \ceil*{\frac{\frac12\abs{\classes_H} + \frac{1}{4}(\abs{\classes_B} + \abs{\classes_{\geq3/4} \setminus (\classes_H \cup \classes_B)}}{\frac12}}\\
	    = \abs{\classes_H} + \ceil*{\frac12\left(\abs{\classes_B} + \abs{\classes_{\geq3/4} \setminus (\classes_H \cup \classes_B)}\right)}
	\]
	and that proves the claim.
\end{proof}

Next, we prove that in $\Oh(n + m\log(m))$ steps it is possible to find the smallest value $T$ with $\max\set{\frac{1}{\machines} p(\jobs),\max_{c \in \classes}p(c)} \leq T \leq \opt$ such that the instance scaled by $1/T$ fulfills the properties from \cref{obs:halfitems,obs:machineBound1,lem:available machines}. 
The algorithms presented in this section will find a schedule with makespan at most $3/2$ for this scaled instance, i.e. the schedule for the original instance will have a makespan of at most $(3/2) T$.

\begin{lemma}\label{lem:scaleValueT}
In $\Oh(n + m\log(m))$ for any given instance $I$, it is possible to find a lower bound $T \leq \opt$ such that for the instance normalized by $1/T$ and the corresponding partition of $\classes$ into $\classes_H, \mathcal{C}_B, \classes_{\geq3/4} \setminus (\classes_H \cup \classes_B)$ and $\classes \setminus \classes_{\geq3/4}$ it holds that
\[\abs{\classes_H} + \max\Set{\abs{\mathcal{C}_B} , \ceil*{\frac12\left(\abs{\classes_B} + \abs{\classes_{\geq3/4} \setminus (\classes_H \cup \classes_B)}\right)}} \leq m.\]
\end{lemma}
\begin{proof}
	By \cref{lower-bounds}, we know that we can set $T \geq \max\set{\frac{1}{\machines} p(\jobs),\max_{c \in \classes}p(c)}$.
	Let $\tilde{p}_{i}$ denote the $(m+1)$-st largest processing time (in a list of processing times containing one entry per job).
	Since each machine can contain at most one job with processing times larger than $\opt/2$, we set
	$T \geq \max\set{\frac{1}{\machines} p(\jobs),\max_{c \in \classes}p(c), \tilde{p}_m + \tilde{p}_{m+1}}$.
	It is possible to find $p_{m+1}$ in $\Oh(n)$ steps, by using the famous median algorithm of Blum et al. \cite{DBLP:journals/jcss/BlumFPRT73}.
	
	Since each class in $\classes_H\cup\classes_B$ contains an item with processing time $\geq 1/2$, only the $m$ classes containing the largest items are candidates for these sets.
	These classes can be found in $\Oh(n)$ by identifying the largest item of each class and comparing it to $p_{m+1}$.
	Similarly the number of classes in $\classes_{\geq3/4}$ is bounded by $(4/3)m$, which can be identified in $\Oh(n)$ by comparing their processing time to $\max\set{\frac{1}{\machines} p(\jobs),\max_{c \in \classes}p(c), \tilde{p}_m + \tilde{p}_{m+1}}$.
	
	After identifying the potential classes, we have to deal with at most $\Oh(m)$ classes.
	For each of these classes there exist three threshold values for $T \in \mathbb{N}$ (i.e., $\lceil\frac{4}{3}(\max_{j\in c}p_j) + 1/3\rceil, 2(\max_{j\in c}p_j)+1$, and $\lceil\frac{4}{3}p(c) + 1/3\rceil$), that would categorize these classes to be no longer in  $\classes_H$, $\classes_B$, and $\classes_{\geq3/4}$, respectively, which after the first two steps can be found in $\Oh(m)$ for all the classes, since they depend on the largest processing time in the class and the total processing time of that class. 
	
	The algorithm can take all these values and sort them by size in $\Oh(m \log(m))$. Via binary search in $\Oh(m \log(m))$, it is possible to find the smallest value $T$ such that $T \geq \max\set{\frac{1}{\machines} p(\jobs),\max_{c \in \classes}p(c), 2p_{m+1}}$ and for the instance normalized by $1/T$ and the corresponding partition into of $\classes$ into $\classes_H, \mathcal{C}_B, \classes_{\geq3/4} \setminus (\classes_H \cup \classes_B)$ and $\classes \setminus \classes_{\geq3/4}$ it holds that
	\[\abs{\classes_H} + \max\Set{\abs{\mathcal{C}_B} , \ceil*{\frac12\left(\abs{\classes_B} + \abs{\classes_{\geq3/4} \setminus (\classes_H \cup \classes_B)}\right)}} \leq m.\qedhere\]
\end{proof}

In the following, we only consider the instance that was scaled by $1/T$. 
We present two Lemmas stating the possibility to partition some classes into two parts that will be scheduled on two different machines.

\begin{lemma}\label{lem:split-classes}
	Let $c \in \classes_{\geq 3/4}$ and $\max_{j\in c}p_j \leq 3/4$. Then $c$ can be partitioned into two parts $\Check{c}$ and $\hat{c}$ with $p(\Check{c}) \leq 1/2$ and $p(\hat{c}) \leq  3/4$ and  $p(\Check{c}) \leq p(\hat{c})$. 
	Furthermore, if $\max_{j\in c}p_j \leq 1/2$, it holds that $p(\Check{c}) \in (1/4,1/2]$ or $p(\hat{c}) \in (1/4,1/2]$.
\end{lemma}
\begin{proof}
	Let $c \in \classes_{\geq 3/4}$ and $\max_{j\in c}p_j \leq 3/4$.
	If $\max_{j\in c}p_j > 1/2$, we set $\hat{c}$ to include the job from $c$ with size bigger than $1/2$ and $\Check{c} = c \setminus \hat{c}$.
	If $\max_{j\in c}p_j \in (1/4,1/2]$, then let $c'$ include a maximal job from $c$ and let $\hat{c},\Check{c}\in \Sett{c',c \setminus c'} $ be distinct such that $p(\Check{c}) \leq p(\hat{c})$.
	Lastly, if $\max_{j\in c}p_j \leq  1/4$, then we construct $c'$ by greedily adding jobs from $c$ to $c'$ until $p(c') > 1/4$ and again define $\hat{c},\Check{c}\in \Sett{c',c \setminus c'} $ to be distinct such that $p(\Check{c}) \leq p(\hat{c})$.
\end{proof}

\begin{lemma}\label{lem:split-classes2}
	Let $c \in \classes$ with $p(c) \in (1/2,3/4)$ and $\max_{j\in c}p_j \leq 1/2$. Then $c$ can be partitioned into two parts $\Check{c}$ and $\hat{c}$ with $p(\Check{c}) \leq p(\hat{c}) \leq 1/2$ and $1/4<p(\hat{c})$.
\end{lemma}
\begin{proof}
	Let $c \in \classes$ with $p(c) \in (1/2,3/4)$ and $\max_{j\in c}p_j \leq 1/2$.
	If $\max_{j\in c}p_j \in (1/4,1/2]$, then let $c'$ include a maximal job from $c$ and let $\hat{c},\Check{c}\in \Sett{c',c \setminus c'} $ be distinct such that $p(\Check{c}) \leq p(\hat{c})$.
	If $\max_{j\in c}p_j \leq  1/4$, then we construct $c'$ by greedily adding jobs from $c$ to $c'$ until $p(c') > 1/4$ and again define $\hat{c},\Check{c}\in \Sett{c',c \setminus c'} $ to be distinct such that $p(\Check{c}) \leq p(\hat{c})$.
	$1/4<p(\hat{c})$ follows directly from the fact that $p(\Check{c}) \leq p(\hat{c})$ and $1/2 < p(\Check{c}) + p(\hat{c})$.
\end{proof}

In the following, we will present two algorithms. 
The first can only handle instances with classes that do not possess an item with processing time larger than $3/4$. 
This algorithm will be used as a subroutine for the second algorithm, which can handle all instances.

\subsection{Algorithm for Instances without Huge Jobs}

	Here we give an algorithm for instances with $\abs{\classes_H}=0$. 
	We assume that the instance was scaled by a value $1/T$ and the classes are categorized as described earlier. 
	The main idea is to repeatedly take combinations of classes with specific parameters which conveniently fill one, two or three machines, without opening additional ones.
	\emph{Fill} in this case means that the average load of \emph{full} machines is in $[1,3/2]$.
	We start with taking two classes with total size in $(1/2,3/4)$ each, as those fill one machine.
	Then we continue with four classes with total size $\geq 3/4$ each, and show how those can be arranged to fill three machines.
	The procedure continues with different combinations of classes until all jobs are scheduled. 
	We show the correctness of the algorithm by arguing that closed machines have on average load of at least $1$, and every scheduled jobs is finished at $3/2$.
	At some point in the algorithm we reach a state where only jobs of classes with total load at most $1/2$ are left.
	Those can be scheduled greedily, by placing full classes on residual machines, until a machine has load at least $1$.
	
	Since we repeatedly have to refer to the jobs which have not been scheduled, we introduce the notation of $\resClasses_X \subseteq \classes_X$ to denote the subset of classes that have not been scheduled at the described step for any class specifier $X$. Note that in the beginning of the algorithm, we have $\resClasses_X = \classes_X$ for all the sets.
	Furthermore, the algorithm will close some of the machines during the construction of the schedule and will not add jobs to closed machines. We denote the set of closed machines as $M_c$. 
	The algorithm is as follows:

\paragraph*{Algorithm: \ANoHuge}
\begin{stepii}
    By applying \cref{lem:split-classes}, partition every class $c \in \classes_{>3/4}$ into two parts $\Check{c},\hat{c}\subseteq c$ with $p(\Check{c}) \leq p(\hat{c}) \leq 3/4$ and $p(\Check{c}) \leq 1/2$.
\end{stepii}
\begin{stepii}\label{step:scheduleClassesWithMediumProcessingTime}
While $\abs{\resClasses_{(1/2,3/4)}} \geq 2$:
	Take $c_1,c_2 \in \resClasses_{(1/2,3/4)}$. 
	Schedule $c_1$ and $c_2$ on one machine such that $c_1$ starts at $0$ and $c_2$ ends at $3/2$.
\end{stepii}
\begin{claim*}
The load of each machine closed in this step is in $(1,3/2)$.
After this step it holds that $\abs{\resClasses_{(1/2,3/4)}} \leq 1$, the partial schedule is feasible, and the total load of closed machines $M_c$ is at least $|M_c|$.
\end{claim*}
\begin{stepii}
\label{step:scheduleClassesWithLargeProcessingTime1}
    While $\abs{\resClasses_{\geq 3/4}} \geq 4$:
	Take $c_1,c_2,c_3,c_4 \in \resClasses_{\geq 3/4}$.
	On the first machine schedule $\hat{c}_1$ and $\hat{c}_2$, such that $\hat{c}_1$ starts at $0$ and $\hat{c}_2$ ends at $3/2$.
	On the second machine schedule $\Check{c}_1$ and $c_3$, such that $\Check{c}_1$ ends at $3/2$ and starts after $1$.
	On the third machine schedule $\Check{c}_2$ and $c_4$, such that $\Check{c}_2$ starts at $0$ and ends before $1/2$ followed by $c_4$, see \cref{fig:4HugeClasses} for an example.
	Close all three machines.
\end{stepii}
\begin{claim*}
After this step $\abs{\resClasses_{(1/2,3/4)}} \leq 1$ and  $\abs{\resClasses_{\geq 3/4}} \leq 3$, the partial schedule is feasible, and the total load of closed machines $M_c$ is at least $|M_c|$.
Furthermore, all scheduled jobs are finished by $3/2$.
\end{claim*}
	
	\begin{figure}[ht]
		\centering
		\begin{subfigure}[b]{0.15\textwidth}
		\tikzset{external/export next=false}
		    \tikz[xscale=0.7,yscale=3.2]{
\defaulthlines{1}
\schedule up [fill=lightgray] {
	{
		.57/$c_1$//,
		.31//opacity=0/,
		.62/$c_2$//
	}/1/
};
}		
	    	\caption{\cref{step:scheduleClassesWithMediumProcessingTime}}
		    \label{fig:scheduleClassesWithMediumProcessingTime}
		\end{subfigure}
		\begin{subfigure}[b]{0.25\textwidth}
		    \tikz[xscale=0.7,yscale=3.2]{
\defaulthlines{3}
\schedule up [fill=lightgray] {
	{
		.68/$\noexpand\hat{c}_1$//,
		.20//opacity=0/,
		.62/$\noexpand\hat{c}_2$//
	}/1/,
	{
		.82/$c_3$//,
		.41//opacity=0/,
		.27/$\noexpand\Check{c}_1$//
	}/1/,
	{
		.30/$\noexpand\Check{c}_2$//,
		.26//opacity=0/,
		.94/$c_4$//
	}/1/
};
}		
	    	\caption{\cref{step:scheduleClassesWithLargeProcessingTime1}}
		    \label{fig:4HugeClasses}
		\end{subfigure}
		\begin{subfigure}[b]{0.2\textwidth}
			\tikz[xscale=0.7,yscale=3.2]{
\defaulthlines{2}
\schedule up [fill=lightgray] {
	{
		.68/$c_3$//,
		.20//opacity=0/,
		.62/$\noexpand\hat{c}_1$//
	}/1/,
	{
		.30/$\noexpand\Check{c}_1$//,
		.26//opacity=0/,
		.94/$c_2$//
	}/1/
};
}
		    \caption{\cref{step:scheduleClassesWithLargeProcessingTime2}}
		    \label{fig:2HugeClasses1Medium}
		\end{subfigure}
		\begin{subfigure}[b]{0.25\textwidth}
		\tikzset{external/export next=false}
			\tikz[xscale=0.7,yscale=3.2]{
\defaulthlines{3}
\schedule up [fill=lightgray] {
	{
		.68/$c$//,
		.35///,
		.42///
	}/1/,
	{
		.22///,
		.12///,
		.37///,
		.28///,
		.25///
	}/1/,
	{
		.42///
	}/1/
};
}
		    \caption{\cref{step:onlyonelargeclass}}
		    \label{fig:onlyonelargeclass}
		\end{subfigure}
		\caption{Examples for \cref{step:scheduleClassesWithMediumProcessingTime,step:scheduleClassesWithLargeProcessingTime1,step:scheduleClassesWithLargeProcessingTime2,step:onlyonelargeclass}}
	\end{figure}
	
\begin{stepii}
\label{step:scheduleClassesWithLargeProcessingTime2}
    If $\abs{\resClasses_{\geq 3/4}} \geq 2$ and  $\abs{\resClasses_{(1/2,3/4)}} = 1$:
	Take $c_1,c_2 \in \resClasses_{\geq 3/4}$ and  $c_3 \in \resClasses_{(1/2,3/4)}$.
	Schedule $c_3$ on the first machine, followed by $\hat{c}_1$ such that it ends at $3/2$.
	Schedule $\Check{c}_1$ on the second machine followed by the jobs from $c_2$ and close both machines, see \cref{fig:2HugeClasses1Medium} for an example.
\end{stepii}
\begin{claim*}
After this step it holds that $\abs{\resClasses_{(1/2,3/4)}} = 0$ and  $\abs{\resClasses_{\geq 3/4}} \leq 3$ or it holds that  $\abs{\resClasses_{(1/2,3/4)}} = 1$ and $\abs{\resClasses_{\geq 3/4}} \leq 1$.
This implies that $\abs{\resClasses_{> 1/2}} \leq 3$ after this step and that $\resClasses_{> 1/2}$ contains at most one class with total processing time less than $3/4$.
Furthermore, the partial schedule is feasible, the total load of closed machines $M_c$ is at least $|M_c|$, and no scheduled job finishes after $3/2$.
\end{claim*}

Depending on the size of $\abs{\resClasses_{> 1/2}}$ the algorithm chooses one of three procedures:

\begin{stepii}\label{step:onlyonelargeclass}
    If $\abs{\resClasses_{>1/2}} \leq 1$:
	Place this class $c$ on one machine. 
	Fill this machine and the residual machines greedily with the residual classes in $\resClasses_{\leq 1/2}$. 
\end{stepii}
\begin{claim*}
After this step it either holds that $2 \leq \abs{\resClasses_{> 1/2}} \leq 3$ or all jobs have been scheduled feasibly with no job finishing after $3/2$.
\end{claim*}
\begin{proof}
In the latter case, we can place all remaining jobs, since there are at least as many open machines as there is open load because before this step we had $p(M_c) \geq |M_c|$. Each opened machine will be filled with load in $[1,3/2]$, since each residual job has a size of at most $1/2$.
\end{proof}
\begin{stepii} \label{step:2greater12}
    If $\abs{\resClasses_{>1/2}} = 2$:
	Let $\resClasses_{>1/2} = \set{c_1,c_2}$ with $p(c_1) \geq p(c_2)$.
	We know that $p(c_1) \geq 3/4$.
		\begin{enumerate}
    		\item If $p(c_2) \leq 3/4$:
    		\begin{enumerate}
    			\item If $p(c_1) + p(c_2) \leq 3/2$:
    			Schedule both on one machine (with $c_1$ starting at $0$ and $c_2$ ending at $3/2$), close it, and continue greedily with the residual jobs.
    			\item If $p(c_1) + p(c_2) > 3/2$:   
    			Place $c_2$ on one machine followed by $\hat{c}_1$ such that $\hat{c}_1$ ends at $3/2$ and close the machine. 
    			Place $\Check{c}_1$  on the next machine and continue greedily with the residual jobs in $\resClasses_{\leq 1/2}$.
    		\end{enumerate}
    		\item If $p(c_2) \geq 3/4$:
    		\begin{enumerate}
    			\item\label{two-large-classes-2a} If $p(\hat{c}_1) + p(\hat{c}_2) \leq 1$:
    			Schedule $c_2$ followed by $\hat{c}_1$ on one machine and close it. 
    			Start $\Check{c}_1$ at $0$ on the next machine and continue greedily with the residual jobs in $\resClasses_{\leq 1/2}$. 
    			\item\label{two-large-classes-2b} If $p(\hat{c}_1) + p(\hat{c}_2) > 1$:
    			Then place $\hat{c}_1$ and $\hat{c}_2$ on one machine such that $\hat{c}_1$ starts at $0$ and $\hat{c}_2$ ends at $3/2$. Place $\check{c}_2$ at the bottom and $\check{c}_1$ at the top of the next machine.			
    			Continue greedily with the residual classes in $\resClasses_{\leq 1/2}$. Start placing them between $\check{c}_2$ and $\check{c}_1$ until the load of that machine is at least $1$ and then continue with the empty machines.
    		\end{enumerate}
		\end{enumerate}

	\begin{figure}[ht]
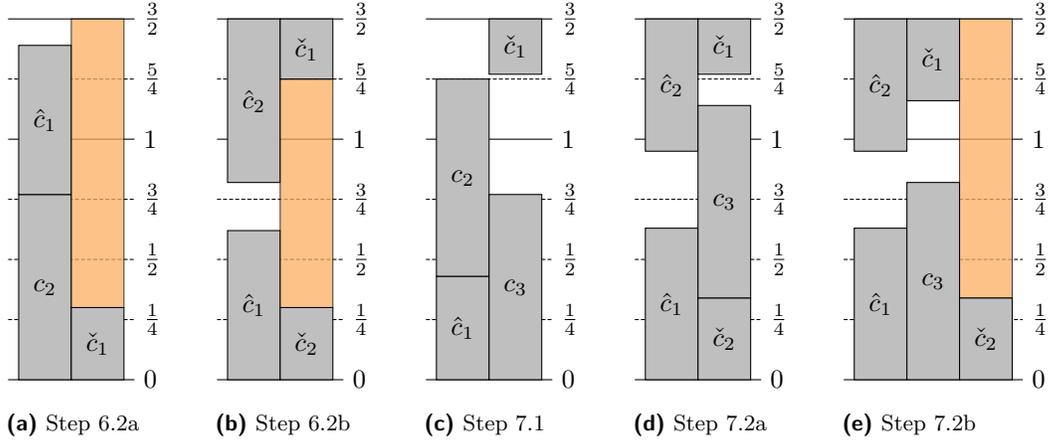

		\centering
		\begin{subfigure}[b]{0.19\textwidth}
    		\tikz[xscale=0.7,yscale=3.2]{
\defaulthlines{2}
\schedule up [fill=lightgray] {
	{
		.77/$c_2$//,
		.62/$\noexpand\hat{c}_1$//,
		.11//opacity=0/
	}/1/,
	{
		.30/$\noexpand\Check{c}_1$//,
		1.20//{opacity=0.75,fill=morange}/
	}/1/
};
}
    		\caption{\Cref{step:2greater12}.\ref{two-large-classes-2a}}
		\end{subfigure}
		\begin{subfigure}[b]{0.19\textwidth}
    		\tikz[xscale=0.7,yscale=3.2]{
\defaulthlines{2}
\schedule up [fill=lightgray] {
	{
		.62/$\noexpand\hat{c}_1$//,
		.20//opacity=0/,
		.68/$\noexpand\hat{c}_2$//
	}/1/,
	{
		.30/$\noexpand\Check{c}_2$//,
		0.95//{opacity=0.75,fill=morange}/,
		.25/$\noexpand\Check{c}_1$//
	}/1/
};
}
    		\caption{\Cref{step:2greater12}.\ref{two-large-classes-2b}}
		\end{subfigure}
		\begin{subfigure}[b]{0.19\textwidth}
    		\tikz[xscale=0.7,yscale=3.2]{
\defaulthlines{2}
\schedule up [fill=lightgray] {
	{
		.43/$\noexpand\hat{c}_1$//,
		.82/$c_2$//
	}/1/,
	{
		.77/$c_3$//,
        {1.5-.77-.23}//opacity=0/,
        .23/$\noexpand\Check{c}_1$//
	}/1/
};
}
    		\caption{\cref{three-large-classes}.\ref{three-large-classes-1}}
		\end{subfigure}
		\begin{subfigure}[b]{0.19\textwidth}
    		\tikz[xscale=0.7,yscale=3.2]{
\defaulthlines{2}
\schedule up [fill=lightgray] {
	{
		.63/$\noexpand\hat{c}_1$//,
        {1.5-.63-.55}//opacity=0/,
		.55/$\noexpand\hat{c}_2$//
	}/1/,
	{
		.34/$\noexpand\Check{c}_2$//,
        .80/$c_3$//,
        {1.5-.34-.23-.80}//opacity=0/,
        .23/$\noexpand\Check{c}_1$//
	}/1/
};
}
    		\caption{\cref{three-large-classes}.\ref{three-large-classes-2a}}
		\end{subfigure}
  		\begin{subfigure}[b]{0.19\textwidth}
    		\tikz[xscale=0.7,yscale=3.2]{
\defaulthlines{3}
\schedule up [fill=lightgray] {
	{
		.63/$\noexpand\hat{c}_1$//,
        {1.5-.63-.55}//opacity=0/,
		.55/$\noexpand\hat{c}_2$//
	}/1/,
	{
        .82/$c_3$//,
        {1.5-.34-.82}//opacity=0/,
		.34/$\noexpand\Check{c}_1$//
	}/1/,
 	{
		.34/$\noexpand\Check{c}_2$//,
		1.16//{opacity=0.75,fill=morange}/
	}/1/
};
}
    		\caption{\cref{three-large-classes}.\ref{three-large-classes-2b}}
		\end{subfigure}
		\caption{Examples for \Cref{step:2greater12} and \Cref{three-large-classes}. Orange blocks represent space for residual classes.}
		\label{fig:2BigClasses}
	\end{figure}
\end{stepii}
\begin{claim*}
After this step it either holds that $\abs{\resClasses_{> 1/2}} = 3$ or all jobs have been scheduled feasibly with no job finishing after $3/2$.
\end{claim*}
\begin{proof}
We will prove the latter case. If $p(c_2) \leq 3/4$, the load of the machines that contains either $c_1$ and $c_2$ or only $c_2$ and $\hat{c}_1$ has a load in $(1,3/2]$. Each residual class (or part of a class) has a total processing time of at most $1/2$. 
Furthermore, up to this step, it holds that $p(M_c) \geq |M_c|$. 
As a consequence, greedily scheduling the residual classes starting with $\check{c}_1$ is possible.

If, on the other hand, $p(c_2) > 3/4$ holds, the machine containing $c_2$ and $\hat{c}_1$ (or $\hat{c}_2$ and $\hat{c}_1$ respectively) has a total load of at least $1$ in either case, and placing $\check{c}_1$ (and $\check{c}_2$) as described does not provoke an overlapping of two jobs requiring the same resource (see \cref{fig:2BigClasses}).
Furthermore, the machine containing $c_2$ and $\hat{c}_1$ (or $\hat{c}_2$ and $\hat{c}_1$ respectively) has a total load of at most $3/2$ since $p(\Check{c}_2)+p(\hat{c}_1) + p(\hat{c}_2)\leq 1/2 + 1$ if $p(\hat{c}_1) + p(\hat{c}_2)\leq 1$ and $p(\hat{c}_i) \leq 3/4$ for $i\in\Sett{1,2}$.
The residual classes again can be scheduled greedily.
This is easy to see in the case $p(\hat{c}_1) + p(\hat{c}_2) \leq 1$ and otherwise we have $p(\check{c}_2) + p(\check{c}_1) \in[0,1)$ and hence the remaining gap has a size of at least $1/2$. 
Since all remaining classes have total load of at most $1/2$ it is possible to greedily add such classes until the total load of that machine is at least $1$ or all remaining classes have been placed.
\end{proof}

\begin{stepii}\label{three-large-classes}
    If $\abs{\resClasses_{>1/2}} = 3$:
	Then $\resClasses_{>1/2} = \resClasses_{\geq 3/4}$.
	Let $\resClasses_{\geq 3/4} = \set{c_1,c_2,c_3}$.
	\begin{enumerate}
		\item\label{three-large-classes-1} If there exists an $i \in \set{1,2,3}$ such that $\hat{c}_i \leq 1/2$:
		Let w.l.o.g. $\hat{c}_1 \leq 1/2$.
		On the first machine schedule $\hat{c}_1$ followed by all the jobs from $c_2$.
		On the next machine schedule all the jobs from $c_3$ and the job $\Check{c}_1$ such that it ends at $3/2$ and close both machines.
        Greedily schedule the jobs in $\resClasses_{\leq 1/2}$ on the non-closed machines.
		\item If $\hat{c}_i > 1/2$ for all $i \in \set{1,2,3}$:
		Place $\hat{c}_1$ and $\hat{c}_2$ on one machine such that $\hat{c}_1$ starts at $0$ and $\hat{c}_2$ ends at $3/2$.
		\begin{enumerate}
			\item\label{three-large-classes-2a} If $p(\Check{c}_1) + p(\Check{c}_2) + p(c_3) \leq 3/2$:
			On the next machine place $\Check{c}_2$ followed by $c_3$ and $\Check{c}_1$ and let $\Check{c}_1$ end at $3/2$.
			Close both machines.
			\item\label{three-large-classes-2b} If $p(\Check{c}_1) + p(\Check{c}_2) + p(c_3) > 3/2$:
			Then w.l.o.g. $p(\Check{c}_1) > 1/4$ and we place $c_3$ and $\Check{c}_1$ on the next machine, such that $\Check{c}_1$ ends at $3/2$. 
			Close both machines.
			On the next machine place $\Check{c}_2$ such that it starts at $0$. 
		\end{enumerate}
		Greedily schedule the jobs in $\resClasses_{\leq 1/2}$ on the non-closed machines.
	\end{enumerate}
\end{stepii}
\begin{claim*}
After this step, all scheduled jobs are finished by $3/2$ and the schedule is feasible.
\end{claim*}
\begin{proof}
Note that the two machines containing the classes $c_1$, $c_2$, and $c_3$ (or $c_1$, $\hat{c}_2$, and $c_3$ respectively) have a total load of at least $2$. As a consequence, all machines $M_c$ closed to this point have a load of at least $|M_c|$. As a consequence there residual load fits on the residual machines. When greedily scheduling the classes each machine is overloaded by at most $1/2$, since each residual class has a processing time of at most $1/2$.
\end{proof}

\begin{lemma}
\label{lem:ANoHuge}
    Given an instance $I=(m,\jobs,\classes)$ that does not contain a huge job, the algorithm \ANoHuge finds a schedule with makespan at most $\frac{3}{2}T$, where $T = \max\set{\frac{1}{\machines} p(\jobs),\max_{c \in \classes}p(c), \tilde{p}_m + \tilde{p}_{m+1}}$.
\end{lemma}

\subsection{Algorithm for the General Case}

Now we present the above-mentioned algorithm that can handle any instance of the problem and uses the previous algorithm in a subroutine.
More specifically, this algorithm places all classes which contain a huge job on a separate machine and fills those machines with jobs from other classes.
This is done by working through different combinations of classes until we reach a point where we can handle the remaining classes and machines as a separate problem instance, at which point the previous algorithm is used.
As before we assume that the instance is scaled by a value $1/T$ and the classes are categorized as described earlier.

We keep the following invariant of the remaining instance over the whole algorithm.

\begin{invariant}
The total load of unscheduled jobs and jobs placed on open machines is bounded by the number of open machines (open machines are all machines not explicitly closed) and in each step the cardinality of the set of unused machines $\bar{M}_u$ is at least 
\[\abs{\bar{M}_u} \geq \max\set{\abs{\bar{\classes}_B}, \ceil{(\abs{\bar{\classes}_B} + \abs{\bar{\classes}_{\geq3/4} \setminus ({\classes}_H \cup {\classes}_B)})/2}}.\]
\end{invariant}

\subparagraph*{Algorithm: \AThreeHalf}
\begin{stepiii}
	Combine specific jobs of the same class into one job. The simplification is done as follows: Iterate all classes $c \in \classes$
	\begin{itemize}
	    \item If $c \in \classes_H$ combine all jobs in $c$ to one huge job.
		\item Else if $p(c) > 3/4$ partition it into parts $\hat{c}$ and $\Check{c}$ with $p(\Check{c}) \leq p(\hat{c}) \leq 3/4$ and $p(\Check{c}) \leq 1/2$.
		Introduce for each part a new job with processing time $p(\hat{c})$ and $p(\Check{c})$, see \cref{lem:split-classes}. 
		\item Else if $c \in \classes_{(1/2,3/4)} \cap \classes_B$:
		partition it into $\hat{c}$ and $\Check{c}$, such that $\hat{c}$ contains the largest job and $\Check{c}$ contains the rest.
		\item Else if $c \in \classes_{(1/2,3/4)} \setminus \classes_B$ partition it into parts $\hat{c}$ and $\Check{c}$ with $p(\Check{c}) \leq p(\hat{c}) \leq 1/2$, see \cref{lem:split-classes2}.
		\item Else if $p(c) \leq 1/2$ introduce one job of size $p(c)$.
	\end{itemize}
\end{stepiii}
\begin{claim*}
    This partition is feasible and every solution for this simplified instance, will still be a solution for the original instance.
\end{claim*}
\begin{stepiii}
    For each $c \in \classes_H$: Open one new machine and assign class $c$ to it. Let $M_H$ be the set of these opened machines. Close all the machines that have load exactly $1$. Denote by $\bar{M}_H$ the set of currently open machines containing a class from $\classes_H$.
\end{stepiii}
\begin{claim*}
	 After this step, there are $\abs{\bar{M}_H}$ many open machines with load in $(3/4,1)$,  $\abs{\resClasses_H}=0$. For the residual empty machines $\bar{M}_u$ it holds that
	 \[\abs{\bar{M}_u} \geq \max\set{\abs{\bar{\classes}_B}, \ceil{(\abs{\bar{\classes}_B} + \abs{\bar{\classes}_{\geq3/4} \setminus ({\classes}_H \cup {\classes}_B)})/2}} \text{\ \ and \ \ } p(\bar{M}_H) + p(\resClasses) \leq \abs{\bar{M}_u} + \abs{\bar{M}_H}.\] 
\end{claim*}	
\begin{stepiii}
    Assign classes $c_s$ with $p(c_s) \leq 1/2$ greedily to machines $\bar{M}_H$ and close each machine with load at least $1$. Continue until either no machines in $\bar{M}_H$ with load less than $1$ is left, or no class with load at most $1/2$ is left.
    If $\abs{\bar{M}_H} = 0$, continue with \ANoHuge on the residual instance.
\end{stepiii}
\begin{claim*}
    After this step 
    either all jobs are scheduled feasibly 
    or it holds that $\abs{\bar{M}_H} \geq 1$ and  $\abs{\resClasses_{\leq 1/2}} = 0$. 
    Furthermore, the partial schedule is feasible, all scheduled jobs are finished by $3/2$ and for the residual empty machines $\bar{M}_u$ it holds that
    \[\abs{\bar{M}_u} \geq \max\set{\abs{\resClasses_B}, \ceil{(\abs{\resClasses_B} + \abs{\resClasses_{\geq3/4} \setminus  \classes_B})/2}} \text{\ \ and \ \ } p(\bar{M}_H) + p(\resClasses) \leq \abs{\bar{M}_u} + \abs{\bar{M}_H}.\]
\end{claim*}	
\begin{proof}
Since we only close machines with load at least one in this step and did not open any new machine, the invariant on the number of unused machines is trivially true. 
Hence, if we have used \ANoHuge on the residual instance, by \cref{lem:ANoHuge} it generates a schedule with makespan at most $3/2$ because $p(\resClasses) \leq |\bar{M}_u|$ at that point and no class was scheduled partially.
\end{proof}
\begin{stepiii}
\label{step:pairMachinesMHwithJobs2}
	While $\abs{\bar{M}_H} \geq 2$ and $\abs{\resClasses_{(1/2,3/4)} \setminus \resClasses_B} \geq 1$:
	Take $m_1,m_2 \in \bar{M}_H$, $c \in \resClasses_{(1/2,3/4)} \setminus \classes_B$.
	Shift the huge job on $m_2$ up such that it ends at $3/2$ and starts at or after $1/2$.
	Schedule $\hat{c}$ on $m_1$ such that it ends at $3/2$, schedule $\Check{c}$ on $m_2$ starting at $0$ and close both machines, see \cref{fig:pairMachinesMHwithJobs2}.
    If $\abs{\bar{M}_H} = 0$, continue with \ANoHuge on the residual instance.
	\begin{figure}[ht]
		\centering
		\begin{subfigure}[b]{0.3\textwidth}
  \tikzset{external/export next=false}
		\tikz[xscale=0.7,yscale=3.2]{
\defaulthlines{2}
\schedule up [fill=lightgray] {
	{
		.78///,
		{1.5-.78-.37}//opacity=0/,
		.37/$\noexpand\hat{c}_1$//
	}/1/,
	{
		.30/$\noexpand\Check{c}_1$//,
		.26//opacity=0/,
		.94///
	}/1/
};
\node at (.5,-.1) {$m_1$};
\node at (1.5,-.1) {$m_2$};
}
		\caption{\Cref{step:pairMachinesMHwithJobs2}}
		\label{fig:pairMachinesMHwithJobs2}
		\end{subfigure}
		\hspace{1ex}
		\begin{subfigure}[b]{0.3\textwidth}
  \tikzset{external/export next=false}
		\tikz[xscale=0.7,yscale=3.2]{
\defaulthlines{2}
\schedule up [fill=lightgray] {
	{
		.78///,
		{1.5-.78-.37}//opacity=0/,
		.37/$\noexpand\Check{c}$//
	}/1/,
	{
		.40/$\noexpand\hat{c}$//,
		.38//opacity=0/,
		.72/$b$//
	}/1/
};
\node at (.5,-.1) {$m_1$};
\node at (1.5,-.1) {$m_2$};
}
		\caption{\Cref{step:HugeJobWithMediumJob1}}
		\label{fig:HugeJobWithMediumJob1}
		\end{subfigure}
		\hspace{1ex}
		\begin{subfigure}[b]{0.3\textwidth}
  \tikzset{external/export next=false}
		\tikz[xscale=0.7,yscale=3.2]{
\defaulthlines{3}
\schedule up [fill=lightgray] {
	{
		.78///,
		{1.5-.78-.37}//opacity=0/,
		.37/$\noexpand\Check{c}_1$//
	}/1/,
	{
		.40/$\noexpand\Check{c}_2$//,
		.32//opacity=0/,
		.78///
	}/1/,
	{
		.52/$\noexpand\hat{c}_1$//,
		.51//opacity=0/,
		.47/$\noexpand\hat{c}_2$//
	}/1/
};
\node at (.5,-.1) {$m_1$};
\node at (1.5,-.1) {$m_2$};
\node at (2.5,-.1) {$m_3$};
}
		\caption{\Cref{step:HugeJobAndHugeClass}}
		\label{fig:HugeJobAndHugeClass}
		\end{subfigure}
		\caption{Examples for \Cref{step:pairMachinesMHwithJobs2}, \Cref{step:HugeJobWithMediumJob1}, and \Cref{step:HugeJobAndHugeClass}}
	\end{figure}
\end{stepiii}
\begin{claim*}
    After this step either all jobs are scheduled feasibly or one of the following two conditions holds: $\abs{\bar{M}_H} =1$ and $\abs{\resClasses_{\leq 1/2}} = 0$,
    or $\abs{\bar{M}_H} \geq 2$ and $\abs{\resClasses \setminus (\classes_B \cup \resClasses_{\geq3/4})} = 0$.
    Furthermore, the partial schedule is feasible, all scheduled jobs are finished by $3/2$ and all machines not in $\bar{M}_H$ are either closed or empty and for the residual empty machines $\bar{M}_u$ it holds that
    \[\abs{\bar{M}_u} \geq \max\set{\abs{\resClasses_B}, \ceil{(\abs{\resClasses_B} + \abs{\resClasses_{\geq3/4} \setminus  \classes_B})/2}} \text{\ \ and \ \ } p(\bar{M}_H) + p(\resClasses) \leq \abs{\bar{M}_u} + \abs{\bar{M}_H}.\]
\end{claim*}
\begin{proof}
Note that we have not opened any other machine in this step, hence the lower bound on $\abs{\bar{M}_u}$ is trivially true.
The total load of $m_1,m_2$ and $c$ is at least $2 \cdot 3/4 + 1/2 = 2$.
Hence in each of these steps, we close two machines but also reduce the residual load by at least $2$, proving the upper bound on the residual load.
Hence, if we have used \ANoHuge on the residual instance, by \cref{lem:ANoHuge} it generates a schedule with makespan at most $3/2$ because $p(\resClasses) \leq |\bar{M}_u|$ at that point and no class was scheduled partially.
\end{proof}
\begin{stepiii}\label{step:OneHugeMach}
    If $\abs{M_H} = 1$: 
	\begin{itemize}	
	\item If there exists $c \in \resClasses \setminus \classes_B$:
		Choose $c'  \in \set{\hat{c},\Check{c}}$ with $c' \in (1/4,1/2]$. 
		Schedule $c'$ on the last open machine $m_0$.
		Use \ANoHuge to schedule the residual instance, including the job $c'' \in c \setminus c'$.
		"Rotate" the load on $m_0$, such that $c'$ does not overlap with $c''$.
	\item If $\resClasses \setminus \classes_B$ is empty:
	Assign all the residual classes to an individual machine.
	\end{itemize}
\end{stepiii}
\begin{claim*}
    After this step all jobs have been scheduled feasibly or
    $\abs{\bar{M}_H} \geq 2$ and 
    $\abs{\resClasses \setminus (\classes_B \cup \resClasses_{\geq3/4})} = 0$.
    Additionally the partial schedule is feasible, all scheduled jobs are finished by $3/2$ and for the residual empty machines $\bar{M}_u$ it holds that
    \[\abs{\bar{M}_u} \geq \max\set{\abs{\resClasses_B}, \ceil{(\abs{\resClasses_B} + \abs{\resClasses_{\geq3/4} \setminus  \classes_B})/2}} \text{\ \ and \ \ } p(\bar{M}_H) + p(\resClasses) \leq \abs{\bar{M}_u} + \abs{\bar{M}_H}.\]
\end{claim*}    
\begin{proof}   
First consider the case that $\resClasses \setminus \classes_B \not = \emptyset$. 
We know that such a required $c'$ exists. 
This is given by \Cref{lem:split-classes2} and \cref{lem:split-classes} for classes in $\resClasses_{(1/2,3/4)} \setminus \classes_B$ and $\resClasses_{\geq 3/4}\setminus \classes_B$, respectively.
The residual instance will be scheduled with the algorithm for instances without huge jobs. 
This generates a feasible schedule, since all machines that are non empty before the start of this subroutine have load at least $1$. Furthermore only class $c$ is partially scheduled and the load on $m_0$ can be rotated, such that $\hat{c}$ and $\check{c}$ do not overlap.
This rotation is always possible: 
The residual job $c''$ of the class is smaller than $3/4$ and will therefore be scheduled consecutively by \ANoHuge.
No matter when $c''$ gets scheduled, before or after it will be a large enough gap that fits $c'$, since $c'$ got scheduled starting at $0$ (or ending at $3/2$ after the rotation). 
Therefore, a correct rotation is possible.

In the case that $\resClasses \setminus \classes_B = \emptyset$,
we put the residual classes to individual machines. This is possible since only classes in $\classes_B$ are left and the number of residual machines is at least $\abs{\resClasses_B}$.
\end{proof}	

\begin{stepiii}
\label{step:HugeJobWithMediumJob1}
    While  $\abs{\bar{M}_H} \geq 1$, $\abs{\resClasses_{(1/2,3/4)} \cap\classes_B} \geq 1$, and $\abs{\resClasses_{\geq3/4}} \geq 1$:
	Take $m_1 \in \bar{M}_H$, $b \in \resClasses_{(1/2,3/4)} \cap\classes_B$ and $c  \in \resClasses_{\geq 3/4}$.
	Open one new machine $m_2$.
	Schedule $\Check{c}$ on $m_1$ such that it ends at $3/2$.
	Schedule $\hat{c}$ on $m_2$ such that it starts at $0$ and ends before $3/4$.
	Schedule $b$ at $m_2$ such that it ends at $3/2$, see \cref{fig:HugeJobWithMediumJob1}.
	Close both machines.
	If $\abs{\bar{M}_H} = 0$, continue with \ANoHuge on the residual instance.
\end{stepiii}
\begin{claim*}
    After this step all jobs are scheduled feasibly or 
    $\abs{\bar{M}_H} \geq 1$ and $\abs{\resClasses \setminus (\resClasses_{(1/2,3/4)} \cap\classes_B)}=0$ or $\abs{\resClasses \setminus \resClasses_{\geq3/4} }=0$.
    Furthermore, all jobs are scheduled feasibly in this step, all scheduled jobs are finished by $3/2$ and for the residual empty machines $\bar{M}_u$ it holds that
    \[\abs{\bar{M}_u} \geq \max\set{\abs{\resClasses_B}, \ceil{(\abs{\resClasses_B} + \abs{\resClasses_{\geq3/4} \setminus  \classes_B})/2}} \text{\ \ and \ \ } p(\bar{M}_H) + p(\resClasses) \leq \abs{\bar{M}_u} + \abs{\bar{M}_H}.\]
\end{claim*}
\begin{proof}
	Note that we open one more machine in each iteration of the step. 
	This machine has to exist, since in each of these steps, we have $\abs{\resClasses_B} + \abs{\resClasses_{\geq3/4} \setminus  \classes_B} \geq 2$.
	In this step, we have reduced $\abs{\resClasses_B} + \abs{\resClasses_{\geq3/4} \setminus  \classes_B}$ by $2$ and $\abs{\resClasses_B}$ at least by $1$.
	Hence there still have to exist $\max\set{\abs{\resClasses_B}, \ceil{(\abs{\resClasses_B} + \abs{\resClasses_{\geq3/4} \setminus \classes_B})/2}}$ unused machines.
	In each iteration of this step, we close two machines but also reduce the residual load by at least $3/4 + 1/2 +3/4 = 2$, proving the upper bound on the residual load.
	
Hence, if we have used \ANoHuge on the residual instance, by \cref{lem:ANoHuge} it generates a schedule with makespan at most $3/2$ because $p(\resClasses) \leq |\bar{M}_u|$ at that point and no class was scheduled partially.
\end{proof}

\begin{stepiii}
    If $\abs{\resClasses_{(1/2,3/4)} \cap\classes_B} \not =0$, open one machine for each of these classes.
\end{stepiii}
\begin{claim*}
    After this step all jobs are feasibly scheduled or it holds that $\abs{\bar{M}_H} \geq 1$ and 
    all residual classes have a total processing time of at least $3/4$, all scheduled jobs are finished by $3/2$ and for the residual empty machines $\bar{M}_u$ it holds that
    \[\abs{\bar{M}_u} \geq \max\set{\abs{\resClasses_B}, \ceil{(\abs{\resClasses_B} + \abs{\resClasses_{\geq3/4} \setminus  \classes_B})/2}} \text{\ \ and \ \ } p(\bar{M}_H) + p(\resClasses) \leq \abs{\bar{M}_u} + \abs{\bar{M}_H}.\]
\end{claim*}	
\begin{proof}
Note that if $\abs{\resClasses_{(1/2,3/4)} \cap\classes_B} \not =0$ this set is the only set containing unscheduled classes. 
Since we still have $\abs{\bar{M}_u} \geq \abs{\resClasses_B}$ unused machines, we can feasibly open one machine for each of these classes and are done.
\end{proof}
\begin{stepiii}
\label{step:HugeJobAndHugeClass}
    While $\abs{\bar{M}_H} \geq 2$ and $\abs{\resClasses_{\geq3/4}} \geq 2$: Take $m_1,m_2 \in \bar{M}_H$, $c_1,c_2 \in \resClasses_{\geq3/4}$ starting with the classes in $\resClasses_B$.
	Shift all jobs on $m_2$ to the top, such that the last job ends at $3/2$.
	Schedule $\Check{c}_{1}$ on $m_1$ as one block that ends at $3/2$ and all the jobs from $\Check{c}_{2}$ as one block on $m_2$ that starts at $0$. 
	Open one more machine $m_3$ where we start the jobs from $\hat{c}_{1}$ at $0$ and let the last job from $\hat{c}_{2}$ end at $3/2$, see \cref{fig:HugeJobAndHugeClass}.
	Close all three machines $m_1, m_2, m_3$.
	If $\abs{\bar{M}_H} = 0$, continue with \ANoHuge on the residual instance.

\end{stepiii}
\begin{claim*}
    After this step all jobs are scheduled or it holds that 
    either $\abs{\bar{M}_H} = 1$ or $\abs{\resClasses_{\geq3/4}} \leq 1$.
    Furthermore $\abs{\resClasses \setminus \resClasses_{\geq3/4}}=0$ and in each iteration the partial schedule is feasible, all scheduled jobs are finished by $3/2$ and for the residual empty machines $\bar{M}_u$ it holds that
    \[\abs{\bar{M}_u} \geq \max\set{\abs{\resClasses_B}, \ceil{(\abs{\resClasses_B} + \abs{\resClasses_{\geq3/4} \setminus  \classes_B})/2}} \text{\ \ and \ \ } p(\bar{M}_H) + p(\resClasses) \leq \abs{\bar{M}_u} + \abs{\bar{M}_H}.\]
\end{claim*}
\begin{proof}
In each of these steps, no two jobs from the same class overlap.
Note that we open one more machine in each iteration of the step. 
This machine has to exist, since $\abs{\resClasses_B} + \abs{\resClasses_{\geq3/4} \setminus \classes_B} \geq 2$ before this step.
Since all remaining classes have load at least $3/4$ it holds that $\abs{\resClasses_B} = \abs{\resClasses_B} \cap \resClasses_{\geq3/4}$.
Therefore, if $\abs{\resClasses_B} \neq 0$ we used at least one such class and reduced $\abs{\resClasses_B}$ by at least $1$.
We also reduced $\abs{\resClasses_B} + \abs{\resClasses_{\geq3/4} \setminus \classes_B}$ by 2 and hence there are still $\max\set{\abs{\resClasses_B}, \ceil{(\abs{\resClasses_B} + \abs{\resClasses_{\geq3/4} \setminus \classes_B})/2}}$ unused machines.
Lastly, in each of these steps, we close three machines but also reduce the residual load by at least $3$ (4 classes with processing time at least $3/4$ each), proving the upper bound on the residual load.

Hence, if we have used \ANoHuge on the residual instance, by \cref{lem:ANoHuge} it generates a schedule with makespan at most $3/2$ because $p(\resClasses) \leq |\bar{M}_u|$ at that point and no class was scheduled partially.
\end{proof}

\begin{stepiii}
    If $\abs{\bar{M}_H} \geq 2$ or $\abs{\resClasses \setminus \classes_B} = 0$, open one machine for each of the remaining classes.
\end{stepiii}
\begin{claim*}
    After this step either all jobs are scheduled or it holds that $\abs{\bar{M}_H} = 1$, $\abs{\resClasses \setminus \resClasses_{\geq3/4}}=0$, $\resClasses \setminus \classes_B\neq\emptyset$, the partial schedule is feasible, all scheduled jobs are finished by $3/2$, and for the residual empty machines $\bar{M}_u$ it holds that
    \[\abs{\bar{M}_u} \geq \max\set{\abs{\resClasses_B}, \ceil{(\abs{\resClasses_B} + \abs{\resClasses_{\geq3/4} \setminus  \classes_B})/2}} \text{\ \ and \ \ } p(\bar{M}_H) + p(\resClasses) \leq \abs{\bar{M}_u} + \abs{\bar{M}_H}.\]
\end{claim*}
\begin{proof}
Due to the previous steps $\abs{\bar{M}_H} \geq 2$ implies $\abs{\resClasses_{\geq3/4}} \leq 1$, and if $\abs{\resClasses_{\geq3/4}} = 0$ we have already scheduled all the jobs.
Otherwise if $\abs{\resClasses_{\geq3/4}} = 1$, there has to be one unused machine because there are at least $\max\set{\abs{\resClasses_B}, \ceil{(\abs{\resClasses_B} + \abs{\resClasses_{\geq3/4} \setminus \classes_B})/2}}$ unused machines.

If, on the other hand, $\abs{\resClasses \setminus \classes_B} = 0$, we still have $\abs{\bar{M}_u} \geq \abs{\resClasses_B}$ unused machines, we can feasibly open one machine for each of these classes and are done.
\end{proof}

\begin{stepiii}
    If $\abs{\bar{M}_H} = 1$, take $c \in \resClasses \setminus \classes_B$.
	It holds that $p(c) \geq 3/4$ and there exists $c' \in \set{\hat{c},\Check{c}}$ with $p(c') \in (1/4,1/2]$.
	Place $c'$ on $m_0 \in \bar{M}_H$. 
	Continue with \ANoHuge to schedule the residual jobs including the job  $c'' \in c \setminus \set{c'}$.
	Rotate the load on $m_0$ such that $c'$ does not overlap with $c''$.
\end{stepiii}
\begin{claim*}
    After this step all jobs are scheduled feasibly and all scheduled jobs are finished by $3/2$.
\end{claim*}

\begin{proof}
The algorithm for instances without huge jobs, can feasibly finish the schedule with makespan at most $3/2$ by \cref{lem:ANoHuge} since $p(\resClasses) \leq |\bar{M}_u|$ at that point, all non empty machines have load at least $1$ on average, and every class except $c$ is either fully scheduled, or not scheduled at all.
Like in \cref{step:OneHugeMach}, the rotation makes sure that there is no conflict within $c$.
\end{proof}

\begin{lemma}
Given any instance $I = (m,\classes)$ of \problem,  \AThreeHalf produces a feasible schedule with makespan at most $\frac{3}{2}\opt(I)$.
\end{lemma}
\begin{proof}
This is a direct consequence when considering the state after each step of the algorithm.
\end{proof}

The existence and correctness of algorithm \AThreeHalf proofs \cref{thm:ThreeHalfAlgorithm}.

\section{Approximation Schemes}\label{sec:schemes}

In this section, we consider approximation schemes for the problem at hand.
An approximation scheme is an algorithm which is given a parameter $\eps > 0$ and an instance $I$ and computes a feasible solution to $I$ whose objective value is guarantied to differ from the optimum $\opt(I)$ by at most $\eps\opt(I)$.
Such an algorithm is called an efficient polynomial-time approximation scheme (EPTAS), if its running time is $f(\eps)\cdot\abs{I}^{\Oh(1)}$ for some computable function $f$.
We present two results:
\begin{theorem}
There exists an EPTAS for \problem if either the number $m$ of machines is constant or $\floor{\eps m}$ additional machines may be used, i.e., some resource augmentation is allowed. 
\end{theorem}
To achieve these results, we follow a framework that was introduced in \cite{JKMR21} and also used in~\cite{DBLP:conf/spaa/JansenLM20}.
In particular, we consider a simplified version of the problem and prove the existence of a certain well-structured solution with only bounded loss in the objective compared to an optimal solution.
The problem of finding such a solution then can be formulated as an integer program (IP) of a particular form.
This IP can be solved efficiently using n-fold integer programming algorithms.
Furthermore, we guarantee that the solution for the simplified problem can be used to derive a solution for the original one with only little loss in the objective value.
The main challenge lies in the design of the well-structured solution and the proof of its existence.
This also causes the limitations of our result:
A certain group of jobs may cause problems in the respective construction and to deal with them we either use a more fine-grained approach, yielding a polynomial running time if $m$ is constant, or place the respective jobs on (few) additional machines using resource augmentation.

\subsection{Simplification}

We use the standard technique (see \cite{DBLP:journals/jacm/HochbaumS87}) of applying a binary search framework to acquire a makespan guess $T$.
The goal is then to either find a schedule of length $(1+\Oh(\eps)) T$ or correctly report that no schedule of length $T$ exists.
We introduce parameters $\delta$ and $\mu$ and call jobs $j$ \emph{big}, \emph{medium}, or \emph{small}, if $p_j\in(\delta T, T]$, $p_j\in(\mu T, \delta T]$, or $p_j\in(0, \mu T]$, respectively. 
Furthermore, we assume $\eps < 0.5$.

\subparagraph{Choosing the Parameters.}

We set $\mu = \eps^2 \delta$ and choose $\delta$ depending on the instance and on whether we consider the case with a constant number of machines or not.
If $m$ is part of the input, we choose $\delta\in\Sett{\eps,\eps^2,\dots,\eps^{2/\eps^2}}$ such that the following two conditions hold:
\begin{enumerate}
\item The overall size of jobs $j$ with size $p_j \in (\mu T,\delta T]$ is at most $\eps^2 mT$.
\item The overall size of jobs $j$ with size $p_j \leq \delta T$ from classes in which these jobs have overall size in $(\mu T,\delta T]$ is at most $\eps^2 mT$.
\end{enumerate}
If, on the other hand, $m$ is fixed, we choose $\delta\in\Sett{\eps,\eps^2,\dots,\eps^{2m/\eps}}$ such that:
\begin{enumerate}
\item The overall size of jobs $j$ with size $p_j \in (\mu T,\delta T]$ is at most $\eps T$.
\item The overall size of jobs $j$ with size $p_j \leq \delta T$ from classes in which these jobs have overall size in $(\mu T,\delta T]$ is at most $\eps T$.
\end{enumerate}
Such a choice is possible in both cases due to the pigeonhole principle.

\subparagraph{Removing the Medium Jobs for fixed $m$.}

Let $I$ be the input instance and $I_1$ the instance we get if we remove all the medium jobs.
\begin{lemma}\label{lem:PTAS_medium_m_constant}
Let $m$ be a constant.
If there is a schedule with makespan $T'$ for $I$, then there is also a schedule with makespan $T'$ for $I_1$; and if there is a schedule with makespan $T'$ for $I_1$, then there is also a schedule with makespan $T' + \eps T$ for $I$.
\end{lemma}
\begin{proof}
The first implication is obvious.
For the other direction, note that the overall size of the medium jobs is upper bounded by $\eps T$ in this case and hence we can place all of them at the end of the schedule on some arbitrary machine.
\end{proof}

\subparagraph{Removing the Medium Jobs for $m$ Part of Input.} 

Let $I$ be the input instance and $I_1$ the instance we get if we remove all the medium jobs from classes including at most $\eps T$ medium load and the entire classes containing more than $\eps T$ medium load.

\begin{lemma}\label{lem:PTAS_medium_m_part_of_input}
Let $m$ be part of the input.
If there is a schedule with makespan $T'$ for $I$, then there is also a schedule with makespan $T'$ for $I_1$; and if there is a schedule with makespan $T'$ for $I_1$, then there is also a schedule with makespan $T' + \eps T$ for $I$ using at most $\floor{\eps m}$ additional machines.
\end{lemma}
\begin{proof}
The first direction is again obvious. 
For the other direction, we first consider the medium jobs from classes including at most $\eps T$ medium load.
We again place these jobs at the end of the schedule.
In particular, they can be placed using the following greedy approach.
We always place all the respective jobs belonging to the same class on the same machine and hence we may glue them together, i.e., assume that each class only contains one job (of size at most $\eps T$).
The jobs are considered ordered decreasingly by size.
On the current machine, we place the jobs starting at time $T'$ one after another until the placement of the next job would result in a makespan greater than $T' + \eps T$ or until no job is left.
If there are jobs left, we continue on the next machine.
Note that due to the ordering of the jobs, we can guarantee that each machine on which we stopped placing jobs to avoid a makespan greater than $T' + \eps T$, we can guarantee that they did receive a load of at least $0.5 \eps T> \eps^2 T$ (using $\eps < 0.5$).
Since the overall load of the medium jobs is at most $\eps^2 m T$, all the considered jobs can be placed.

Next, we consider the classes including more than $\eps T$ medium load which were removed completely. 
Let $\classes'$ be the set of classes containing more than $\eps T$ medium load and let $p_m(c)$ the corresponding load.
Then $\abs{\classes'} \eps T < \sum_{c\in \classes'}p_m(c) \leq \eps^2 mT $ yielding $\abs{\classes'} < \eps m$.
Hence, we can place these classes on $\floor{\eps m}$ additional machines such that each machine receives exactly one class.
Since we place the entire classes, this cannot cause conflicts.
\end{proof}

\subparagraph{Removing Some Small Jobs.}

A $(T',L')$-schedule is a schedule with makespan at most $T'$ and at least $L'$ idle time throughout the schedule.
Let $I_2$ be the instance we get if we remove all the small jobs from classes in which these jobs have overall size of at most $\delta T$ from $I_1$.
Let $L$ be the overall size of the jobs removed in this step.
We obviously have:
\begin{lemma}\label{lem:PTAS_small_removal}
If there is a schedule with makespan $T'$ for $I_1$, then there is also a $(T',L)$-schedule for $I_2$.
\end{lemma}

\subparagraph{Layered Schedule and Rounded Processing Times.}

Next, we will consider certain well-structured schedules called \emph{layered schedules}.
For some positive number $\xi$, we call a schedule $\xi$-layered, if the processing of each job starts at a multiple of $\xi$.
The time between two such multiples is called a \emph{layer} and the corresponding time on a single machine a \emph{slot}.

In the following we show, that such a layered schedule can be generated, by rounding processing times and fusing jobs.
Let $I_3$ be the instance we get if we round the processing times of the big jobs and replace the remaining small jobs with placeholders.
In particular, let $p'_j = \ceil{p_j/(\eps\delta T)}\eps\delta T$ be the rounded size for each big job $j$.
Furthermore, for each class $c\in\classes$ with $p_s(c)=\sum_{j\in c, p_j \leq \mu T}p_j >\delta T$, we remove the small jobs and introduce $\ceil{p_s(c)/(\eps\delta T)}$ new jobs with processing time $\eps\delta T$ each.
\begin{lemma}\label{lem:PTAS_rounded_sizes_layered}
If there is a $(T',L)$-schedule for $I_2$, then there is also an $\eps \delta T$-layered $((1+ 2 \eps)T',L)$-schedule for $I_3$.
\end{lemma}
\begin{proof}
Consider a $(T',L)$-schedule for $I_2$ and picture each job in a container.
We stretch the schedule by a factor of $(1+ 2 \eps)$, and, in doing so, we move and stretch the containers correspondingly.
The jobs inside the container, however, are not stretched.
Note that each job can be moved inside its container without creating conflicts and we initially move each job to the bottom of its container.
Next, we move each \emph{big} job inside its container up such that it starts at the next layer border and increase its size to the rounded size.
Since each big job $j$ has size at least $\delta T$ its container has size at least $p_j + 2\eps\delta T$ and therefore each job remains inside its container.
At this point, each big job starts and ends at a layer border and has the correct rounded size.

Next, we consider the small jobs.
Let $S$ be the set of small jobs in $I_2$ and $n_c$ the number of placeholders we want to introduce for class $c$, i.e., $n_c = \ceil{p_s(c)/(\eps\delta T)}$.
For each class $c$, we grow the size of the small jobs inside their containers until the overall size of the jobs in $c\cap S$ is equal to $n_c \eps\delta T$. 
This is again possible since $p_s(c)> \delta T$.
We denote the changed processing time of a small job $j$ as $p^*_j$. 
Now, slots in which (parts of) some small job are placed can only contain small load due to the steps we performed for the big jobs.
We will use the initial distribution of the small jobs as a starting point to find a feasible placement of the placeholder jobs in the layers.
To do so, we first need some additional notation.
Let $\layers = \Sett[\big]{\ell\in\mathbb{Z}_{>0} \given (\ell-1) \eps\delta T \leq (1+2\eps)T'}$ be the set of layers;
$\lambda(j,\ell)$ the fraction of job $j\in S$ placed in layer $\ell\in\layers$ (i.e. $\sum_{\ell \in \layers} \lambda_{j,\ell} =1 $);
$p^*_{\ell} = \sum_{j\in S} \lambda_{j,\ell}p^*_j$ the small load placed in layer $\ell$; 
$k_{\ell} =\ceil{ p^*_{\ell} / (\eps\delta T) }$ the rounded up number of slots needed for the overall small load in $\ell$; 
and $\gamma_{c,\ell} = \ceil{\sum_{j\in c\cap S}\lambda_{j,\ell} }\in \Sett{0,1}$ a parameter indicating whether a small job belonging to class $c$ is scheduled in layer $\ell$.
Note that we have $\sum_{\ell \in \layers}\sum_{j\in S\cap c} \lambda(j,\ell) = n_c$
We now construct a flow network with integral capacities for which the given placement of the small jobs yields a maximum flow and utilize flow integrality to find a feasible placement for the placeholder.
A very similar approach was taken in \cite{JKMR21,DBLP:conf/spaa/JansenLM20}.
The flow network is defined as follows and visualized in \cref{fig:flow-network}.
\begin{itemize}
\item There is a source $\alpha$, a sink $\omega$, a node $u_c$ for each class $c$, and a node $v_\ell$ for each layer $\ell$.
\item The source is connected to each class node $u_c$ via an edge $(\alpha, u_c)$ with capacity $n_c$.
\item Each class node $u_c$ is connected to each node $v_{\ell}$ via an edge $(u_c,v_\ell)$ with capacity $\gamma_{c,\ell}$.
\item Each layer node $v_{\ell}$ is connected to the sink via an edge $(v_\ell, \omega)$ with capacity $k_{\ell}$.
\item A maximum flow $f$ is given by $f(\alpha, u_c) = n_c$, $f(u_c,v_\ell) = \sum_{j\in S\cap c} \lambda(j,\ell)$, and $f(v_\ell, \omega) = \sum_{j\in S} \lambda(j,\ell)$.
\end{itemize}
\begin{figure}
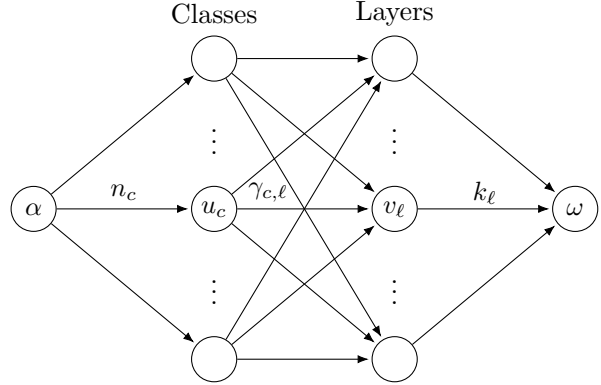

\centering
\tikz[>={Latex[length=1.5mm]}, shorten <= 0pt, shorten >= 1pt]{
\pgfmathsetmacro{\w}{2.4}
\pgfmathsetmacro{\h}{2}
\pgfmathsetmacro{\nw}{0.6}


\node[draw, circle] (source) at (0,0) {$\alpha$};

\node[draw, circle, inner sep=0, minimum size = \nw cm] (class1) at (\w,\h) {};
\node at (\w,0.5*\h) {$\vdots$};
\node[draw, circle, inner sep=0, minimum size = \nw cm] (class2) at (\w,0) {$u_c$};
\node at (\w,-0.5*\h) {$\vdots$};
\node[draw, circle, inner sep=0, minimum size = \nw cm] (class3) at (\w,-\h) {};

\node[draw, circle, inner sep=0, minimum size = \nw cm] (layer1) at (2*\w,\h) {};
\node at (2*\w,0.5*\h) {$\vdots$};
\node[draw, circle, inner sep=0, minimum size = \nw cm] (layer2) at (2*\w,0) {$v_{\ell}$};
\node at (2*\w,-0.5*\h) {$\vdots$};
\node[draw, circle, inner sep=0, minimum size = \nw cm] (layer3) at (2*\w,-\h) {};

\node[draw, circle, inner sep=0, minimum size = \nw cm] (sink) at (3*\w,0) {$\omega$};


\foreach \x in {1,2,3}
{
\draw[->] (source) -- (class\x);
}

\foreach \y in {1,2,3}
{
\foreach \x in {1,2,3}
{
\draw[->] (class\x) -- (layer\y);
}
}

\foreach \x in {1,2,3}
{
\draw[->] (layer\x) -- (sink);
}


\node at (\w,1.3*\h) {\large Classes};
\node at (2*\w,1.3*\h) {\large Layers};
\node at (0.5*\w,0.11*\h) {$n_c$};
\node at (1.3*\w,0.11*\h) {$\gamma_{c,\ell}$};
\node at (2.5*\w,0.11*\h) {$k_{\ell}$};

}
\caption{Flow network used in the construction of a layered schedule.}
\label{fig:flow-network}
\end{figure}
It is easy to check that $f$ is indeed a feasible maximum flow.
Moreover, all capacities are integral and hence there exists an integral flow $f'$ with the same value.
We now remove all the small jobs from the schedule and assign the placeholders into slots according to $f'$.
In particular, we assign a placeholder small job belonging to class $c$ to a slot in layer $\ell$ if $f'(c,\ell) = 1$.
Since slots that originally did receive some small jobs cannot contain any big load, the definition of $k_{\ell}$ together with flow conservation imply that there are enough slots to do so.
Furthermore, this cannot produce conflicts since each layer receives at most one placeholder job of each class and the presence of a big job of a class $c$ in a layer $\ell$ implies $\gamma_{c,\ell} = 0$ and hence prevents the placement of a placeholder.
Hence, we did construct a layered schedule with makespan at most $(1+2\eps)T'$ for $I_3$.
Finally note that the free space is preserved as well since jobs were only increased inside their containers after the stretching step. 
\end{proof}

\subparagraph{Reinserting the Small Jobs.}

Finally, we discuss the reinsertion of all of the small jobs as well as the use of the original sizes:
\begin{lemma}\label{lem:PTAS_small_reinsertion}
If an $\eps \delta T$-layered $(T',L)$-schedule for $I_3$, then there is also a schedule with makespan $(1+\eps)T' + \eps T$ for $I_1$.
\end{lemma}
\begin{proof}
We first discuss the insertion of the small jobs starting with the ones from classes in which small jobs have overall size in $(\mu T,\delta T]$.
Due to the choice of the medium jobs, the respective jobs have overall size of at most $ \eps^2 m T$ if $m$ is part of the input or $\eps T$ if $m$ is constant.
In the former, we can use a simple greedy procedure similar to the one in the proof of \cref{lem:PTAS_medium_m_part_of_input} to place them at the end of the schedule, and in the latter, we can just place all of them at the end of the schedule on an arbitrary machine.
In either case, the objective value grows by at most $\eps T$.

Next, we stretch the schedule by a factor of $(1+ \eps)$ increasing the sizes of the jobs, layers, and free space accordingly.
Now, each placeholder small job has a size of $(1+ \eps)\eps\delta T = \eps\delta T + \mu T$.
Hence, if we remove a placeholder belonging to a class $c$ and greedily place original small jobs of class $c$ into the respective slot, we can guarantee that at least a load of $\eps\delta T$ is placed (unless all of the small jobs of the class already have been placed). 
Finally, we place the jobs from classes in which the small jobs have an overall size of at most $\mu T$.
If there is a big job in the same class, we fix one of them and decrease its size by $\mu T$ (this is a smaller decrease than the increase due to the stretching step) and place the small jobs in the freed space.
Else, we can place them greedily in the free slots placing all the jobs of the same class in the same slot (again utilizing that each free slot was increased by $\mu T$ in the stretching step).
This cannot produce conflicts since these classes do not contain any big jobs.
In a last step we reduce the sizes of the big jobs to their original ones. 
\end{proof}

\subsection{Integer Program}

To find a layered schedule, we utilize an IP approach. 
The corresponding IP is essentially a \emph{module configuration IP} as introduced in \cite{JKMR21} but, for the sake of simplicity, we diverge from the notation in the respective work.
Considering \cref{lem:PTAS_medium_m_constant,lem:PTAS_medium_m_part_of_input,lem:PTAS_small_removal,lem:PTAS_rounded_sizes_layered}, we set $T' = (1+2\eps)T$ and search for an $\eps\delta T$-layered $(T',L)$-schedule for $I_3$.
We introduce some notation.
Let $\layers = \Sett[\big]{\ell\in\mathbb{Z}_{>0} \given (\ell-1) \eps\delta T \leq (1+2\eps)T'}$ be the set of layers, $P$ be the set of distinct processing times in $I_3$, and $n^{(c)}_p$ the number of jobs of size $p$ in class $c$ for each $p\in P$ and $c\in\classes$.
Furthermore, we define a (time) \emph{window} as a pair $(\ell,p) \in \layers \times P$ of a starting layer $\ell$ and a processing time $p$, and a \emph{configuration} $K$ as a selection of windows $\Sett{0,1}^{\windows}$ such that no two conflicting windows are chosen, i.e.,  $\sum_{(\ell, p) \in \windows_{\ell'}} K_{\ell, p} \leq 1$ for each layer $\ell'$.
The set of configurations is denoted as $\confs$, the set of windows as $\windows$, and the the set of windows intersecting layer $\ell$ as $\windows_{\ell}$.
A window $(\ell,p)$ intersects $p/(\eps\delta T)$ many succeeding layers starting with layer $\ell$.
\begin{observation}\label{obs:bounds_ip}
We have $\abs{P} \in \Oh(1/(\eps\delta))$, $\abs{\layers} \in \Oh(1/(\eps\delta))$, $\abs{\windows}\in \Oh(1/(\eps\delta)^2)$, and $\abs{\confs} \in 2^{\Oh(1/(\eps\delta)^2)}$.
\end{observation} 
\begin{proof}
Due to the rounding, the processing times are multiples of $\eps\delta T$ and upper bounded by $T + \eps\delta $.
Hence, we have $P\in \Oh(1/(\eps\delta))$ and essentially the same argument yields $\layers \in \Oh(1/(\eps\delta))$ which directly implies $\abs{\windows}\in \Oh(1/(\eps\delta)^2)$ which in turn yields $\abs{\confs} \in 2^{\Oh(1/(\eps\delta)^2)}$.
\end{proof}
In the IP, we have a variable $x_{K}\in\Sett{0,\dots,m}$ for each $k\in \confs$, a variable $y^{(c)}_{\ell,p}\in\Sett{0,\dots,n}$ for each class $c\in\classes$ and window $(\ell,p)\in\windows$, as well as the following constraints:
\begin{align}
\sum_{K\in\mathcal{K}}x_{K} & = m & \label{eq:ip_Nconfigs=Nmachines} \\ 
\sum_{K\in\mathcal{K}} K_{\ell,p} x_{K} & = \sum_{c\in\classes}y^{(c)}_{\ell,p} & \forall (\ell,p)\in\windows \label{eq:ip_confs_cover_windows}\\
\sum_{\ell\in \layers} y^{(c)}_{\ell,p} & = n^{(c)}_p & \forall c\in\classes,p\in P \label{eq:ip_windows_cover_jobs}\\
\sum_{(\ell',p)\in\windows_{\ell}} y^{(c)}_{\ell',p} & \leq 1 & \forall c\in\classes, \ell\in\layers \label{eq:ip_no_conflicts_window_choice}
\end{align}
The variables $y^{(c)}_{\ell,p}$ are used to reserve time windows for the placement of jobs belonging to class $c$, \eqref{eq:ip_windows_cover_jobs} guarantees that the correct number is chosen, and due to \eqref{eq:ip_no_conflicts_window_choice} placing the respective jobs in the windows will not create conflicts.
Furthermore, the variables $x_{K}$ are used to chose $m$ configurations (due to \eqref{eq:ip_Nconfigs=Nmachines}).
Each such configuration corresponds to a scheduling pattern on one of the $m$ machines.
In particular, a configuration is by definition a selection of non-overlapping time windows and in \eqref{eq:ip_confs_cover_windows} we make sure that these configurations cover the selected windows.
Hence, it is easy to construct a solution for the IP given a $\eps \delta T$-layered $(T',L)$-schedule and vice-versa yielding:
\begin{lemma}
There exists an $\eps \delta T$-layered $(T',L)$-schedule for $I_3$, if and only if the above IP is feasible.
\end{lemma}

\subsection{Algorithm and Analysis}

Summing up, we use a binary search framework to get the makespan guess $T$, perform the simplification steps described in \cref{lem:PTAS_medium_m_constant,lem:PTAS_medium_m_part_of_input,lem:PTAS_small_removal,lem:PTAS_rounded_sizes_layered}, formulate and solve the described IP, construct a schedule from the IP solution, and transform it into a schedule for the original instance using the steps described in \cref{lem:PTAS_medium_m_constant,lem:PTAS_medium_m_part_of_input,lem:PTAS_rounded_sizes_layered,lem:PTAS_small_reinsertion}.
For the given makespan guess $T$, we thus find a schedule with makespan at most $(1+\eps)(1+2\eps)T + 2\eps T = (1+ \Oh(\eps))T$ or, if the IP is not feasible, correctly report that a schedule with makespan $T$ does not exist.

Regarding the running time, it is easy to see that the critical step lies in solving the IP, since all the other ones mostly involve simple changes of the instance and fast greedy procedures that obviously run in polynomial time. 
Hence, we take a closer look at the IP and again essentially apply the approach introduced in \cite{JKMR21}, i.e., solving it via n-fold integer programming.

\subparagraph{N-fold Integer Programming.}

A (generalized) n-fold IP is an IP of the form
\[
    \min\Sett{c^Tx \given \mathcal{A} x = b, \ell \leq x\leq u, x\in\mathbb{Z}^{Nt}}
\]
with $N,r,s,t \in \mathbb{Z}_{> 0}$, $c,\ell,u\in\mathbb{Z}^{Nt}$, $b\in\mathbb{Z}^{r + Ns}$, $A_i \in \mathbb{Z}^{r\times t}$ and $B_i \in \mathbb{Z}^{s\times t}$ for each $i\in [N]$, as well as
\[
\mathcal A =
\begin{pmatrix}
A_1	& \dots & A_N \\
B_1	& \dots & 0 \\
\vdots	& \ddots & \vdots \\
0 & \dots & B_N
\end{pmatrix}.
\]
The n-fold IPs and variants thereof are intensively studied in ongoing research and there has been a series of better and better algorithms presented in recent years.
For an overview of most of these developments, we refer to the extensive work \cite{DBLP:journals/corr/abs-1904-01361}.
We will employ the most recent result for this family of problems by Cslovjecsek et al.~\cite{DBLP:conf/soda/CslovjecsekEHRW21}:
\begin{theorem}\label{thm:solving_n-fold}
The $N$-fold integer programming problem with $\Delta$ the maximum absolute value occurring in $\mathcal A$ can be solved in time $2^{\Oh(rs^2)} (rs\Delta)^{\Oh(r^2s + s^2)} (Nt)^{1+o(1)}$.
\end{theorem}
We call $N$ the number of blocks, $r$ the number of global constraints, $s$ the number of local constraints, and $t$ the number of block variables.

\subparagraph{Application to the Present IP.}

We have to slightly change our IP in order to bring it into the above form.
In particular, we copy the variables $x_{K}$ such that each variable is present $\abs{\classes}$ many times but do not use any of these variables except for the original copy (hence the constraints are not changed).
Furthermore, we introduce slack variables for \cref{eq:ip_no_conflicts_window_choice} to transform the constraint into an equation. 
After performing these steps the number of blocks $N$ is equal to $\abs{\classes}$, the number of block variables is given by $\abs{\confs} + \abs{\windows} + \abs{\layers}$, and the number of global and local constraints by $\abs{\windows} + 1$ (\cref{eq:ip_Nconfigs=Nmachines,eq:ip_confs_cover_windows}) and $\abs{P} + \abs{\layers}$, respectively.
Furthermore, the parameter $\Delta$ is equal to $1$.
Hence, the IP can be solved in time $f(1/\eps)\cdot\abs{I}^{\Oh(1)}$ for some doubly exponential function $f$.

\section{Inapproximability Results}\label{sec:inapprox}
We consider the case in which each job may need more than a single resource.
Let us assume that we have a set $\resources$ of resources and each job $j$ needs some subset $\resources(j)$ in order to be processed.
The classes then correspond to subsets of resources $R\subseteq\resources$ with $\jobs(R) = \Sett{j\in\jobs \given \resources(j) =R}$.
We can adapt an APX-hardness result from \cite{DBLP:journals/scheduling/EvenHKR09} by recreating their conflict graph with resources.
This is done by creating a resource $r_e$ per edge $e=\{u,v\}$ and letting jobs $u$ and $v$ require that resource.
This reduction needs 2 machines, job sizes in ${1,2,3,4}$ but roughly as many distinct resources per job, as there are jobs.
Subsequently, we give a new unrelated reduction for an instance of the problem with a constant bound on the number of distinct resources per job.

\begin{theorem}
    There is no $5/4 - \varepsilon$ approximation algorithm with $\varepsilon > 0$ for the \problem with multiple resources per job if $P \neq NP$.
    This holds true, even if no job needs more than 3 resources ($\forall j \in \jobs: \abs{\resources(j)}\leq 3$) and all jobs have processing time 1, 2 or 3 ($\forall j \in \jobs: p(j) \in \{1,2,3\}$).
    Furthermore, this also holds when the number of machines is unlimited.
\end{theorem}

\begin{proof}
    We show this by giving a reduction from the NP-hard \textsc{Monotone 3-Sat-(2,2)} problem \cite{DBLP:journals/dam/DarmannD21}, which is a satisfiability problem with the following restrictions:
    The boolean formula is in 3CNF, each clause contains either only unnegated or negated variables and each literal appears in exactly 2 clauses (and every variable in exactly 4 clauses).
    Note here, that we only use the bounded occurrence of literals, not the monotony.
    
    In the following we write that two (or more) jobs $j$ and $j'$ "share a resource $r$", which means that $r\in \resources(j)$ and $r\in \resources(j')$ and for all other jobs $j^*$, $j^*\neq j$ and $j^*\neq j'$, $r\notin \resources(j^*)$.
    
    Let $\phi$ be the given formula and $\mathfrak{C}$, $\mathfrak{X}$ the sets of clauses and variables in $\phi$, respectively. 
    We start by creating a dummy structure that we can anchor jobs to by using shared resources.
    Create $\abs{\mathfrak{C}}$ many pairs of dummy jobs $j^A_i$, $j^a_i$ with $p(j^A_i) = 3, p(j^a_i) = 1$, which share a unique resource $A_i$.
    Furthermore, $j^a_i$ and $j^A_{i+1}$ share a unique resource $A_{i\rightarrow i+1}$.
    Create $\abs{\mathfrak{X}}$ many pairs of dummy jobs $j^b_i$, $j^B_i$ with $p(j^b_i) = p(j^B_i) = 2$, which share a unique resource $B_i$.
    Furthermore, $j^B_i$ and $j^b_{i+1}$ share a unique resource $B_{i\rightarrow i+1}$.
    Lastly, $j^a_{\abs{\mathfrak{C}}}$ and $j^b_{1}$ share a unique resource $A_{\rightarrow B}$
    
    For every $x_i \in \mathfrak{X}$ create three variable jobs $j_{x_i}$, $j_{\bar{{x_i}}}$ and $j_{d{x_i}}$ which all share a resource $X_{x_i}$, moreover $j_{d{x_i}}$ and $j^B_{i}$ share a resource $B_{x_i}$.
    
    For every $c_i \in \mathfrak{C}$ with $c_i= \{ x_1^{c_i}, x_2^{c_i}, x_3^{c_i} \} $ create four clause jobs $j_{x_1}^{c_i}$, $j_{x_2}^{c_i}$, $j_{x_3}^{c_i}$ and $j_{d}^{c_i}$ which all share a resource $C_{c_i}$.
    Furthermore, $j_{d}^{c_i}$ and $j^A_i$ share a resource $A_{c_i}$, while $j_{x_1}^{c_i}$ and the corresponding negated or unnegated variable job $j_{x_k}$ or $j_{\bar{{x_k}}}$ share a resource $V_{x_k}^{c_i}$.
    Note here, that each variable job $j_{x_k}$ and $j_{\bar{{x_k}}}$ shares two \emph{unique} resources with the corresponding clause jobs.
    
    Finally, we set the number of machines to $2\abs{\mathfrak{C}} + 2\abs{\mathfrak{X}}$ (remark: we could also give an unlimited number of machines, as the resources limit the number of concurrently usable machines either way).
    
	\begin{figure}[ht]
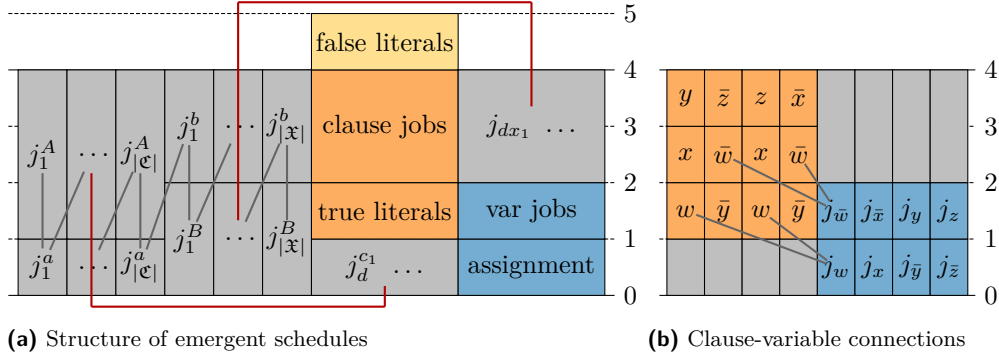

		\centering
		\begin{subfigure}[b]{0.6\textwidth}
    		\tikz[xscale=0.65,yscale=3]{
\hlines{12}{
    {5/4}/$5$//,
    {4/4}/$4$/solid/,
    {3/4}/$3$//,
    {2/4}/$2$//,
    {1/4}/$1$//,
    0/$0$/solid/
};

\schedule up [fill=lightgray] {
    {
        .25/$j_1^a$//,
        .75/$j_1^A$//
    }/1/,
    {
        .25/\dots//,
        .75/\dots//
    }/1/,
    {
        .25/$j^a_{|\noexpand\mathfrak{C}|}$//,
        .75/$j^A_{|\noexpand\mathfrak{C}|}$//
    }/1/,
    {
        .5/$j_1^B$//,
        .5/$j_1^b$//
    }/1/,
    {
        .5/\dots//,
        .5/\dots//
    }/1/,
    {
        .5/$j_{|\noexpand\mathfrak{X}|}^B$//,
        .5/$j_{|\noexpand\mathfrak{X}|}^b$//
    }/1/,
    {
        .25/$j_d^{c_1}$ \dots //,
        .25/true literals/fill=morange/,
        .5/clause jobs/fill=morange/,
        .25/false literals/fill=myellow/
    }/3/,
    {
        .25/assignment/fill=mblue/,
        .25/var jobs/fill=mblue/,
        .5/$j_{dx_1}$ \dots //
    }/3/
};

\newcommand{\mycolor}{black!60!white}

\draw[thick, shorten <=0.25cm, shorten >=0.25cm, \mycolor] (0.5,0.125) -- (0.5,0.625);
\draw[thick, shorten <=0.25cm, shorten >=0.25cm, \mycolor] (0.5,0.125) -- (1.5,0.625);
\draw[thick, shorten <=0.25cm, shorten >=0.25cm, \mycolor] (1.5,0.125) -- (2.5,0.625);
\draw[thick, shorten <=0.25cm, shorten >=0.25cm, \mycolor] (2.5,0.125) -- (2.5,0.625);

\draw[thick, shorten <=0.25cm, shorten >=0.25cm, \mycolor] (2.5,0.125) -- (3.5,0.75);

\draw[thick, shorten <=0.25cm, shorten >=0.25cm, \mycolor] (3.5,0.25) -- (3.5,0.75);
\draw[thick, shorten <=0.25cm, shorten >=0.25cm, \mycolor] (3.5,0.25) -- (4.5,0.75);
\draw[thick, shorten <=0.25cm, shorten >=0.25cm, \mycolor] (4.5,0.25) -- (5.5,0.75);
\draw[thick, shorten <=0.25cm, shorten >=0.25cm, \mycolor] (5.5,0.25) -- (5.5,0.75);

\newcommand{\myaltcolor}{black!30!red}

\draw[thick, shorten <=0.25cm, \myaltcolor] (1.5,0.625) -- (1.5,-0.05);
\draw[thick, \myaltcolor] (1.5,-0.05) -- (7.5,-0.05);
\draw[thick, shorten >=0.25cm, \myaltcolor] (7.5,-0.05) -- (7.5,0.125);

\draw[thick, shorten <=0.25cm, \myaltcolor] (4.5,0.25) -- (4.5,1.30);
\draw[thick, \myaltcolor] (4.5,1.30) -- (10.5,1.30);
\draw[thick, shorten >=0.25cm, \myaltcolor] (10.5,1.30) -- (10.5,0.755);
}		
    		\caption{Structure of emergent schedules}
    		\label{fig:redu1}
		\end{subfigure}
		\begin{subfigure}[b]{0.35\textwidth}
			\centering
    		\tikz[xscale=0.5,yscale=3]{
\hlines{8}{
    {4/4}/$4$/solid/,
    {3/4}/$3$//,
    {2/4}/$2$//,
    {1/4}/$1$//,
    0/$0$/solid/
};

\schedule up [fill=lightgray] {
    {
        .25///,
        .25/$w$/fill=morange/,
        .25/$x$/fill=morange/,
        .25/$y$/fill=morange/
    }/1/,
    {
        .25///,
        .25/$\noexpand\bar{y}$/fill=morange/,
        .25/$\noexpand\bar{w}$/fill=morange/,
        .25/$\noexpand\bar{z}$/fill=morange/
    }/1/,
    {
        .25///,
        .25/$w$/fill=morange/,
        .25/$x$/fill=morange/,
        .25/$z$/fill=morange/
    }/1/,
    {
        .25///,
        .25/$\noexpand\bar{y}$/fill=morange/,
        .25/$\noexpand\bar{w}$/fill=morange/,
        .25/$\noexpand\bar{x}$/fill=morange/
    }/1/,
    {
        .25/$j_w$/fill=mblue/,
        .25/$j_{\noexpand\bar{w}}$/fill=mblue/,
        .5///
    }/1/,
    {
        .25/$j_x$/fill=mblue/,
        .25/$j_{\noexpand\bar{x}}$/fill=mblue/,
        .5///
    }/1/,
    {
        .25/$j_{\noexpand\bar{y}}$/fill=mblue/,
        .25/$j_y$/fill=mblue/,
        .5///
    }/1/,
    {
        .25/$j_{\noexpand\bar{z}}$/fill=mblue/,
        .25/$j_z$/fill=mblue/,
        .5///
    }/1/
};

\newcommand{\mycolor}{black!60!white}

\draw[thick, shorten <=0.15cm, shorten >=0.15cm, \mycolor] (0.5,0.375) -- (4.5,0.125);
\draw[thick, shorten <=0.15cm, shorten >=0.15cm, \mycolor] (2.5,0.375) -- (4.5,0.125);
\draw[thick, shorten <=0.15cm, shorten >=0.15cm, \mycolor] (1.5,0.625) -- (4.5,0.375);
\draw[thick, shorten <=0.15cm, shorten >=0.15cm, \mycolor] (3.5,0.625) -- (4.5,0.375);

}
    		\caption{Clause-variable connections}
    		\label{fig:redu2}
		\end{subfigure}
		\caption{Dummy structure and connection between clause and variable jobs. Grey lines represent a resource each, red lines represent pairwise resources of $j^A_i$ and $j_{d}^{c_i}$ ($j^B_{i}$ and $j_{d{x_i}}$, respectively).}
	\end{figure}
    
    \begin{lemma}
        There is an optimal schedule with makespan 4 if and only if there is a satisfying assignment for the \textsc{Monotone 3-Sat-(2,2)} problem. Otherwise the optimal schedule has a makespan of 5.
    \end{lemma}
    We first show that there is a trivial schedule with makespan 5 for each instance of the resulting scheduling problem.
    Place the dummy jobs as in \Cref{fig:redu1}, after that, for each $j_{d{x_i}}$ place the two corresponding jobs $j_{x_i}$ and $j_{\bar{{x_i}}}$ directly below it (in any order).
    Lastly, for every $j_{d}^{c_i}$ leave the timestep directly above it empty, and place $j_{x_1}^{c_i}$, $j_{x_2}^{c_i}$, $j_{x_3}^{c_i}$ above the empty time step, finishing in timestep 5.
    It should be easy to see, that this is always possible.
    
    Secondly we show how to construct a schedule with makespan 4 if there is a satisfying assignment.
    We again start by placing the dummy jobs as in \Cref{fig:redu1}.
    For each $j_{d{x_i}}$ we place the two corresponding jobs $j_{x_i}$ and $j_{\bar{{x_i}}}$ below it. Now, look at the satisfying assignment for $\phi$, if $x_i$ is true (false) in the assignment, $j_{x_i}$ is placed below (above) $j_{\bar{{x_i}}}$.
    The job corresponding to the true assignment finishes at timestep 1, the other at 2.
    For every $j_{d}^{c_i}$ place $j_{x_1}^{c_i}$, $j_{x_2}^{c_i}$, $j_{x_3}^{c_i}$ above it.
    From the three jobs, choose one of which the corresponding literal evaluates to true in the given assignment to be placed directly above $j_{d}^{c_i}$ (note that there has to be at least one such job, since the assignment satisfies $\phi$).
    This is the only of the three (non-dummy) clause jobs that overlaps the variable jobs placed earlier but the variable job it shares a resource with, was scheduled in the first timestep (see \Cref{fig:redu2}).
    
    Lastly, we show how to construct a satisfying assignment from a schedule with makespan 4. 
    One can verify, that each dummy job in such a schedule has a fixed time window where it has to be scheduled due to the conflicts with other dummy jobs as (we ignore that the whole schedule can be "flipped on its head" since it is equivalent).
    Furthermore, in such a schedule every time interval on every machine must be filled.
    Each pair of $j^A_i$, $j^a_i$ or $j^b_i$, $j^B_i$ occupies an interval of $[0,4]$, $j_{d{x_i}}$ is scheduled in $[0,1]$ and every $j_{d}^{c_i}$ is scheduled in $[2,4]$ (see \Cref{fig:redu1}).
    We count the remaining open slots: $[0,1]$: $\abs{\mathfrak{X}}$, $[1,2]$: $\abs{\mathfrak{X}}+\abs{\mathfrak{C}}$ and $[2,4]$: $\abs{\mathfrak{C}}$.
    The variable jobs $j_{x_i}$ and $j_{\bar{{x_i}}}$ can only be scheduled in $[0,2]$ (due to their dummy job) and can not be scheduled concurrently (due to their shared resource).
    Therefore, for every pair $j_{x_i}$ and $j_{\bar{{x_i}}}$ one job is scheduled in $[0,1]$ and one in $[1,2]$.
    After that the remaining open slots are: $[1,4]$: $\abs{\mathfrak{C}}$.
    Following an analogously argumentation for each triple of clause jobs $j_{x_1}^{c_i}$, $j_{x_2}^{c_i}$, $j_{x_3}^{c_i}$ one job gets scheduled in $[1,2]$, $[2,3]$ and $[3,4]$.
    For every of those in $[1,2]$ the corresponding variable job has to be scheduled in $[0,1]$ (because they share a resource), which gives us, that the variable jobs in $[0,1]$ directly correspond to a satisfying assignment for the original \textsc{Monotone 3-Sat-(2,2)} Problem.
\end{proof}

Remark: Following a similar construction we can show the same inapproximability result for unit jobs and 8 or less resources per job.
Furthermore, it is possible to give a $4/3 - \varepsilon$ inapproximability result for the problem by giving a reduction from the NAE-3SAT problem. That reduction uses unit jobs, but a non-constant number of resources per job.

\section{Conclusion}

In this paper, we did greatly improve the state of the art regarding the approximability of \problem.
There are several interesting avenues emerging for further investigations.
Firstly, there is the question of whether a PTAS for \problem without resource augmentation can be achieved.
It seems plausible that the approximation schemes results of the present work could be further refined to reach this goal.
For the case with only a constant number of machines, on the other hand, an FPTAS is not ruled out at this point.

Moreover, it would be interesting to explore natural extensions of \problem and, in particular, to investigate for which variants approximation schemes may or may not be feasible. 
From the negative perspective, we have already provided initial results in this paper.
We would like to point out one further question in this direction:
Note that \problem can be seen as a special case of scheduling with conflicts where the conflict graph is a cograph. 
This problem is known to be NP-hard already for unit size jobs \cite{DBLP:conf/mfcs/BodlaenderJ93} and it would be interesting to explore inapproximability for arbitrary sizes.
Regarding the design of approximation schemes, on the other hand, variants where the corresponding conflict graph is a particularly simple cograph may be interesting.

Finally, from a broader perspective, it seems interesting to explore the possibilities of N-fold IPs and related concepts \cite{DBLP:journals/corr/abs-1904-01361} for scheduling with additional resources.

\bibliography{bibbib}

\end{document}